%% file: HKRVN15.tex
\documentclass[a4,11pt]{amsart}

\usepackage{cmlgc}
\usepackage{ucs}
\usepackage[english]{babel}
\usepackage{bm,amsfonts,amsmath,amssymb,dsfont}
\usepackage{graphicx}
\usepackage[utf8x]{inputenc}
\usepackage[T1]{fontenc}
\usepackage{enumerate}
\usepackage{mathrsfs} 

\usepackage[babel=true]{csquotes}

\usepackage{pdfsync}

\makeatletter
\def\section{\@startsection{section}{1}\z@{.9\linespacing\@plus\linespacing}%
  {.7\linespacing} {\fontsize{13}{15}\selectfont\scshape\centering}}
\def\paragraph{\@startsection{paragraph}{4}%
  \z@{0.3em}{-.5em}%
  {$\bullet$ \ \normalfont\itshape}}
\makeatother

\newtheorem{theo}{Theorem}[section]
\newtheorem{proposition}[theo]{Proposition}
\newtheorem{lemma}[theo]{Lemma}
\newtheorem{corollary}[theo]{Corollary}

\theoremstyle{definition}

\newtheorem{notation}[theo]{Notation}
\newtheorem{hyp}[theo]{Assumption}
\renewcommand{\geq}{\geqslant}
\renewcommand{\leq}{\leqslant}
\newcommand{\h}{\hbar}
\newcommand{\deriv}[2]{\frac{\partial #1}{\partial #2}}
\newcommand{\norm}[1]{\left\|#1\right\|}
\newcommand{\dd}{\theta}

\theoremstyle{remark}
\newtheorem{rem}[theo]{Remark}

\makeatletter

\@addtoreset{equation}{section}
\makeatother

\usepackage{float}
\usepackage{color}

\definecolor{gr}{rgb}   {0.,   0.69,   0.23 }
\definecolor{bl}{rgb}   {0.,   0.5,   1. }
\definecolor{mg}{rgb}   {0.85,  0.,    0.85}
\definecolor{yl}{rgb}   {0.8,  0.7,   0.}

\definecolor{webred}{rgb}{0.75,0,0}
\definecolor{webgreen}{rgb}{0,0.75,0}
\usepackage[citecolor=webgreen,colorlinks=true,linkcolor=webred]{hyperref}

\newcommand{\R}{\mathbb{R}}
\newcommand{\C}{\mathbb{C}}
\newcommand{\N}{\mathbb{N}}
\newcommand{\bA}{\mathbf{A}}
\newcommand{\bB}{\mathbf{B}}

\newcommand{\bb}{\mathbf{b}}
\newcommand{\bc}{\mathbf{c}}
\newcommand{\bd}{\mathbf{d}}
\newcommand{\f}{\mathbf{f}}

\newcommand{\e}{\mathbf{e}}
\newcommand{\z}{\mathbf{z}}

\newcommand{\dx}{\mathrm{d}}
\newcommand{\dr}{\partial}
\newcommand{\eps}{\varepsilon}
\newcommand{\ess}{\mathsf{ess}}
\newcommand{\spe}{\mathsf{sp}}

\newcommand{\pscal}[2]{\langle #1, #2 \rangle }

\newcommand{\Op}{\operatorname{Op}}

\newcommand{\OO}{\mathcal{O}}
\newcommand{\opweyl}{\Op_{\hbar}^w}
\newcommand{\abs}[1]{\left|#1\right|}
\newcommand{\formel}[1]{[\![#1]\!]}

\begin{document}

\title[Magnetic wells in dimension three]{\large Magnetic wells in
  dimension three} 
\author{B. Helffer}
\address{~\newline B. Helffer: D\'epartement de Math\'ematiques (B\^at. 425), Universit\'e Paris-Sud et CNRS, 91405 Orsay C\'edex (France) \& Universit\'e de Nantes}
\email{bernard.helffer@math.u-psud.fr}
\author{Y. Kordyukov}
\address{~\newline Y. Kordyukov: Institute of Mathematics, Russian Academy of Sciences, 112 Chernyshevsky str. 450008 Ufa (Russia)}
\email{yurikor@matem.anrb.ru}
\author{N. Raymond} 
\address{~\newline N. Raymond: IRMAR (UMR 6625), Universit{\'e}
  de Rennes 1,
  Campus de Beaulieu, 35042 Rennes cedex (France)}
  \email{nicolas.raymond@univ-rennes1.fr}
\author{S. V\~u Ng\d{o}c}
\address{~\newline S. V\~u Ng\d{o}c: IRMAR (UMR 6625), Universit{\'e}
  de Rennes 1,
  Campus de Beaulieu, 35042 Rennes cedex (France) \& Institut Universitaire de France.}
\email{san.vu-ngoc@univ-rennes1.fr}


\begin{abstract}
  This paper deals with semiclassical asymptotics of the
  three-dimensional magnetic Laplacian in presence of magnetic
  confinement. Using generic assumptions on the geometry of the
  confinement, we exhibit three semiclassical scales and their
  corresponding effective quantum Hamiltonians, by means of three
  microlocal normal forms \textit{\`a la Birkhoff}. As a consequence,
  when the magnetic field admits a unique and non degenerate minimum,
  we are able to reduce the spectral analysis of the low-lying
  eigenvalues to a one-dimensional $\hbar$-pseudo-differential
  operator whose Weyl's symbol admits an asymptotic expansion in powers of
  $\hbar^{\frac1 2}$.
\end{abstract}
\maketitle

\section{Introduction}

\subsection{Motivation and context}
The analysis of the magnetic Laplacian $(-i\hbar\nabla-\bA)^2$ in the
semiclassical limit $\hbar\to 0$ has been the object of many
developments in the last twenty years. The existence of discrete
spectrum for this operator, together with the analysis of the
eigenvalues, is related to the notion of ``magnetic bottle'', or
quantum confinement by a pure magnetic field, and has important
applications in physics. Moreover, motivated by investigations of the
third critical field in Ginzburg-Landau theory for superconductivity,
there has been a great attention focused on estimates of the lowest
eigenvalue. In the last decade, it appears that the spectral analysis
of the magnetic Laplacian has acquired a life on its own. For a story
and discussions about the subject, the reader is referred to the
recent reviews \cite{FouHel10, HelKo14, Ray14}. 

In contrast to the wealth of studies exploring the semiclassical
approximations of the Schrödinger operator $-\h^2\Delta+V$, the
\textit{classical} picture associated with the Hamiltonian
$\|p-\bA(q)\|^2$ has almost never been investigated to describe the
\textit{semiclassical} bound states (\textit{i.e.} the eigenfunctions
of low energy) of the magnetic Laplacian. The paper by Raymond and
V{\~u}~Ng{\d{o}}c \cite{RVN13} is to our knowledge the first rigorous
work in this direction. In that paper, which deals with the
two-dimensional case, the notion of magnetic drift, well known to
physicists, is cast in a symplectic framework, and using a
semiclassical Birkhoff normal form (see for instance \cite{Vu06,
  VuCha08, Vu09}) it becomes possible to describe all the eigenvalues
of order $\OO(\hbar)$. Independently, the asymptotic expansion of a
smaller set of eigenvalues was established in \cite{HelKo11, HelKo13b}
through different methods which act directly on the quantum side:
explicit unitary transforms and a Grushin like reduction are used to
reduce the two-dimensional operator to an effective one-dimensional
operator.

The three-dimensional case happens to be much harder.  The only known
results in this case that provide a full asymptotic expansion of a
given eigenvalue concern toy models where the confinement is obtained
by a boundary carrying a Neumann condition on an half space in
\cite{Ray12} or on a wedge in \cite{PoRay12}. In the case of smooth
confinement without boundary, a construction of quasimodes by Helffer
and Kordyukov in \cite{HelKo13} suggests what the expansions of the
low lying eigenvalues could be. But, as was expected by Colin de
Verdière in his list of open questions in~\cite{colin-moyennisation},
extending the symplectic and microlocal techniques to the
three-dimensional case contains an intrinsic difficulty in the fact
that the symplectic form cannot be nondegenerate on the characteristic
hypersurface. The goal of our paper is to answer this question by
fully carrying out this strategy. After averaging the cyclotron
motion, the effect of the degeneracy of the symplectic form can be
observed on the fact that the reduced operator is only partially
elliptic. Hence, the key ingredient will be a separation of scales via
the introduction of a new semiclassical parameter for only one part of
the variables. These semiclassical scales are reminiscent of the three
scales that have been exhibited in the classical picture in the large
field limit, see~\cite{benettin-sempio, cheverry14}.  They are also
related to the Born-Oppenheimer type of approximation in quantum
mechanics (see for instance \cite{BO27, Martinez07}). In fact, in a
partially semiclassical context and under generic assumptions, a full
asymptotic expansion of the first magnetic eigenvalues (and the
corresponding WKB expansions) has been recently established in any
dimension in the paper by
Bonnaillie-No\"el--H\'erau--Raymond~\cite{BHR14}.

\subsection{Magnetic geometry}
Let us now describe the geometry of the problem. The configuration
space is
$$\R^3=\{q_{1}\e_{1}+q_{2}\e_{2}+q_{3}\e_{3},\quad q_{j}\in\R, \quad j=1,2,3\},$$
where $(\e_{j})_{j=1,2,3}$ is the canonical basis of $\R^3$.  The
phase space is
$$\R^6=\{(q,p)\in\R^3\times\R^3\}$$
and we endow it with the canonical $2$-form
\begin{equation}\label{omega0}
  \omega_{0}=\dx p_{1}\wedge \dx q_{1}+\dx p_{2}\wedge \dx q_{2}+\dx p_{3}\wedge \dx q_{3}.
\end{equation}
We will use the standard Euclidean scalar product
$\langle\cdot,\cdot\rangle$ on $\R^3$ and $\|\cdot\|$ the associated
norm. In particular, we can rewrite $\omega_{0}$ as
$$\omega_{0}((u_{1}, u_{2}),(v_{1},v_{2}))=\langle v_{1}, u_{2}\rangle-\langle v_{2}, u_{1}\rangle,\quad \forall u_{1}, u_{2}, v_{1}, v_{2}\in\R^3.$$
The main object of this paper is the magnetic Hamiltonian, defined for
all $(q,p)\in\R^6$ by
\begin{equation}\label{H}
  H(q,p)=\|p-\bA(q)\|^2,
\end{equation}
where $\bA\in\mathcal{C}^\infty(\R^3, \R^3)$.

Let us now introduce the magnetic field. The vector field $\bA=(A_{1},
A_{2}, A_{3})$ is associated (via the Euclidean structure) with the
following $1$-form
$$\alpha=A_{1}\dx q_{1}+A_{2}\dx q_{2}+A_{3}\dx q_{3}$$
and its exterior derivative is a $2$-form, called magnetic $2$-form
and expressed as
$$\dx \alpha=(\dr_{1}A_{2}-\dr_{2}A_{1})\dx q_{1}\wedge \dx q_{2}+(\dr_{1}A_{3}-\dr_{3}A_{1})\dx q_{1}\wedge \dx q_{3}+(\dr_{2}A_{3}-\dr_{3}A_{2})\dx q_{2}\wedge \dx q_{3}\,.$$
The form $\dx\alpha$ may be identified with a vector field. If we let:
$$\bB=\nabla\times\bA=(\dr_{2}A_{3}-\dr_{3}A_{2}, \dr_{3}A_{1}-\dr_{1}A_{3}, \dr_{1}A_{2}-\dr_{2}A_{1})=(B_{1}, B_{2}, B_{3}),$$
then, we can write
\begin{equation}\label{B}
  \dx \alpha=B_{3}\dx q_{1}\wedge \dx q_{2}-B_{2}\dx q_{1}\wedge \dx q_{3}+B_{1}\dx q_{2}\wedge \dx q_{3}.
\end{equation}
The vector field $\bB$ is called the magnetic field. Let us notice
that we can express the $2$-form $\dx \alpha$ thanks to the magnetic
matrix
$$
M_{\bB}= \left(\begin{array}{ccc}
    0&B_{3}&-B_{2} \\
    -B_{3}&0&B_{1} \\
    B_{2}&-B_{1}&0
\end{array}\right).
$$
Indeed we have
\begin{equation}\label{MB}
  \dx \alpha(U,V)=\langle U,M_{\bB}V\rangle=\langle U,V\times\bB\rangle=[U, V,\bB],\quad\forall (U,V)\in\R^3\times \R^3,
\end{equation}
where $[\cdot,\cdot,\cdot]$ is the canonical mixed product on
$\R^3$. We note that $\bB$ belongs to the kernels of $M_{\bB}$ and
$\dx\alpha$.

An important role will be played by the characteristic hypersurface
$$\Sigma=H^{-1}(0),$$ 
which is the submanifold defined by the
parametrization:
\[
\R^3 \ni q \mapsto j(q):= (q,\bA(q)) \in \R^3\times\R^3.
\]
We may notice the relation between $\Sigma$, the symplectic structure and the magnetic field in the following relation
\begin{equation}\label{lemma1}
  j^*\omega_0 =\dx \alpha\,,
\end{equation}
where $\dx \alpha$ is defined in~\eqref{B}.

\subsection{Confinement assumptions and discrete spectrum}
This paper is devoted to the semiclassical analysis of the discrete
spectrum of the magnetic Laplacian
$\mathcal{L}_{\hbar,\bA}:=(-i\h\nabla_{q} - \bA(q))^2$, which is the semiclassical Weyl quantization of $H$ (see \eqref{Weyl}). This means that we will consider that
$\hbar$ belongs to $(0,\hbar_{0})$ with $\hbar_{0}$ small enough.

Let us recall the assumptions under which discrete spectrum actually exist. In two dimensions, with a non vanishing magnetic field, a standard estimate (see \cite{AHS, CFKS87}) gives
$$
\hbar \int_{\mathbb R^2} |B(q)| |u(q)|^2 \dx q \leq \langle \mathcal{L}_{\hbar,\bA} u\,|\, u \rangle\,,\, \forall u \in \mathcal{C}_0^\infty(\mathbb R^2)\,.
$$
Except in special cases when some components of the magnetic field have constant sign, this is no more the case in higher dimension (see \cite{Duf83}). We should impose a control of the oscillations of $\bB$ at infinity. Under this condition, we get a similar estimate at the price of a small loss.
This kind of estimate actually follows from an analysis developed in \cite{HelMo96}. Let us define 
$$b(q):=\|\bB(q)\|.$$
Let us now state the confining assumptions under which we will constantly work in this paper.
\begin{hyp}\label{hyp1-2}
We consider the case of $\mathbb R^3$ and assume
\begin{equation}\label{hyp1}
b(q)\geq b_{0}:=\inf_{q\in\R^3}b(q)>0\,,
\end{equation}
and the existence of a constant $C>0$ such that
\begin{equation}\label{hyp2}
\|\nabla \bB (q)\| \leq C\, (1 + b(q))\,,\, \forall q\in \mathbb R^3\,.
\end{equation}
\end{hyp}
Under Assumption \ref{hyp1-2}, it is proven in \cite[Theorem 3.1]{HelMo96} that there exist $h_0>0$ and $C_0 >0$ such that, for all $\hbar\in (0,h_0)$,
\begin{equation}\label{lbB}
\hbar (1- C_0 \hbar^\frac 14)\int_{\mathbb R^3}b(q) |u(q)|^2 \dx q \leq \langle \mathcal{L}_{\hbar,\bA} u\,|\, u \rangle\,,\, \forall u \in \mathcal{C}_0^\infty(\mathbb R^3)\,.
\end{equation}

As a corollary, using Persson's theorem (see \cite{Persson60}), we obtain that the bottom of the essential spectrum is  asymptotically above $\hbar b_1$, where
$$
b_1 :=  \liminf_{|q| \rightarrow +\infty} b(q).
$$
More precisely, under Assumption \ref{hyp1-2}, there exist $h_0>0$ and $C_0 >0$ such that, for all $\hbar\in (0,h_0)$,
\begin{equation}\label{spessentiel}
 \mathfrak{s}_{\ess} (\mathcal L_{\hbar,\bA} ) \subset [\hbar b_1 (1-C_0\hbar^\frac 14), +\infty).
\end{equation}
\begin{hyp}\label{hyp3-4}
We assume that
\begin{equation}\label{hyp3}
0 < b_0 < b_1\,.
\end{equation}
Moreover we will assume that there exists a point $q_{0}\in\R^3$ and $\eps>0$, $\tilde\beta_{0}\in(b_{0}, b_{1})$ such that
\begin{equation}\label{hyp4}
\{ b(q) \leq \tilde\beta_0 \}\subset D(q_{0},\eps),
\end{equation}
where $D(q_{0},\eps)$ is the Euclidean ball centered at $q_{0}$ and of
radius $\eps$. For the rest of the article we let
$\beta_0\in(b_0,\tilde\beta_0)$. Without loss of generality, we can assume that $q_{0}=0$ and that $\bA(0)=0$ (which can be obtained with a change of gauge).
\end{hyp}
Note that Assumption \ref{hyp3-4} implies that the minimal value of $b$ is attained inside $D(q_{0},\eps)$.

All along this paper, we will strengthen the assumptions on the nature of the point $q_{0}$. At some stage of our investigation, $q_{0}$ will be the unique minimum of $b$. Note in particular that \eqref{hyp4} is satisfied as soon as $b$ admits a unique and non degenerate minimum.

\subsection{Informal description of the results}
Let us now informally walk through the main results of this paper. We
will assume (as precisely formulated in \eqref{hyp3}-\eqref{hyp4}) that the magnetic field does not vanish and is confining.

Of course, for eigenvalues of order $\OO(\hbar)$, the corresponding
eigenfunctions are microlocalized in the semi-classical sense near the
characteristic manifold $\Sigma$ (see for instance \cite{Ro87,
  Z13}). Moreover the confinement assumption implies that the
eigenfunctions of $\mathcal{L}_{\hbar,\bA}$ associated with
eigenvalues less that $\beta_{0}\hbar$ enjoy localization estimates
\textit{\`a la} Agmon. Therefore we will be reduced to investigate the
magnetic geometry locally in space near a point $q_{0}=0\in\R^3$
belonging to the confinement region and which, for notational
simplicity, we may assume to be the origin.

Then, in a neighborhood of $(0,\bA(0))\in\Sigma$, there exist
symplectic coordinates $(x_{1},\xi_{1}, x_{2},\xi_{2}, x_{3},\xi_{3})$
such that $\Sigma=\{x_{1}=\xi_{1}=\xi_{3}=0\}$ and $(0,\bA(0))$ has coordinates $0\in\R^6$. Hence $\Sigma$ is
parametrized by $(x_{2},\xi_{2},x_{3})$.
\subsubsection{First Birkhoff form}
In these coordinates suited for the magnetic geometry, it is possible
to perform a semiclassical Birkhoff normal form and microlocally
unitarily conjugate $\mathcal{L}_{\hbar,\bA}$ to a first normal form
$\mathcal{N}_{\hbar}=\Op_{\hbar}^w\left(N_{\hbar}\right)$ with an
operator valued symbol $N_{\hbar}$ depending on $(x_{2},\xi_{2},
x_{3},\xi_{3})$ in the form
$$N_{\hbar}=\xi_{3}^2+b(x_{2},\xi_{2},x_{3})\mathcal{I}_{\hbar}+f^{\star}(\hbar,\mathcal{I}_{\hbar}, x_{2},\xi_{2}, x_{3},\xi_{3})+\OO(|\mathcal{I}_{\hbar}|^{\infty},|\xi_{3}|^\infty).$$ 
where $\mathcal{I}_{h}=\hbar^2D_{x_{1}}^2+x_{1}^2$ is the first encountered harmonic oscillator and where $(\hbar,I, x_{2},\xi_{2}, x_{3},\xi_{3})\mapsto f^\star (\hbar,I, x_{2},\xi_{2}, x_{3},\xi_{3})$ satisfies, for $I\in(0, I_{0})$,
$$|f^\star (\hbar,I, x_{2},\xi_{2}, x_{3},\xi_{3})|\leq C\left(|I|^{\frac{3}{2}}+|\xi_{3}|^{3}+\hbar^{\frac{3}{2}}\right).$$
Since we wish to describe the spectrum in a spectral window containing
at least the lowest eigenvalues, we are led to replace
$\mathcal{I}_{\hbar}$ by its lowest eigenvalue $\hbar$ and thus, we
are reduced to the two-dimensional pseudo-differential operator
$\mathcal{N}_{\hbar}^{[1]}=\Op_{\hbar}^w\left(N_{\hbar}^{[1]}\right)$
where
$$N^{[1]}_{\hbar}=\xi_{3}^2+b(x_{2},\xi_{2},x_{3})\hbar+f^{\star}(\hbar,\hbar, x_{2},\xi_{2}, x_{3},\xi_{3})+\OO(\hbar^{\infty},|\xi_{3}|^\infty).$$ 

\subsubsection{Second Birkhoff form}
If we want to continue the normalization, we shall assume a new
non-degeneracy condition (the first one was the positivity of
$b$). 

Now we assume that, for any $(x_2,\xi_2)$ in a neighborhood of $(0,0)$, the function $x_{3}\mapsto b(x_{2},\xi_{2}, x_{3})$ admits a unique and non-degenerate minimum denoted by $s(x_{2},\xi_{2})$. Then, by using a new symplectic transformation in
order to center the analysis at the partial minimum
$s(x_{2},\xi_{2})$, we get a new operator
$\underline{\mathcal{N}}_{\hbar}^{[1]}$ whose Weyl symbol is in the
form
$$\underline{N}_{\hbar}^{[1]}=\nu^2(x_{2},\xi_{2})(\xi^2_{3}+\hbar x_{3}^2)+\hbar b(x_{2},\xi_{2}, s(x_{2},\xi_{2}))+\mathsf{remainders},$$
with 
\begin{equation}\label{eq.nu}
\nu(x_{2},\xi_{2}) = (\tfrac{1}{2}\partial_3^2 b(x_2,\xi_2,s(x_2,\xi_2)))^{1/4}
\end{equation}
and where the remainders have been properly normalized to be at least
formal perturbations of the second harmonic oscillator
$\xi_{3}^2+\hbar x_{3}^2$. Since the frequency of this oscillator is
$\hbar^{-\frac{1}{2}}$ in the classical picture, we are naturally led
to introduce the new semiclassical parameter $$
h=\hbar^{\frac{1}{2}}$$ and the new
impulsion $$\xi=\hbar^{\frac{1}{2}}\tilde\xi\,$$ so
that $$\Op_{\hbar}^w\left(\xi_{3}^2+\hbar
  x_{3}^2\right)=h^2\Op_{h}^w\left(\tilde\xi_{3}^2+x_{3}^2\right).$$
We therefore get the $h$-symbol of
$\underline{\mathcal{N}}_{\hbar}^{[1]}$:
$$\underline{\mathsf{N}}_{h}^{[1]}=h^2\nu^2(x_{2},h\tilde\xi_{2})(\tilde\xi^2_{3}+ x_{3}^2)+h^2 b(x_{2},h\tilde\xi_{2}, s(x_{2},h\tilde\xi_{2}))+\mathsf{remainders}.$$
We can again perform a Birkhoff analysis in the space of formal series given by $\mathscr{E}= \mathscr{F}[\![x_3,\tilde \xi_3, h]\!]$ where $\mathscr{F}$ is a space of symbols in the form $c(h, x_{2},h\tilde\xi_{2})$. We get the new operator $\mathfrak{M}_{h}=\Op_{h}^w\left(\mathsf{M}_{h}\right)$,
with
\begin{multline*}
\mathsf{M}_{h}=h^2 b(x_2,h \tilde\xi_2,s(x_2, h\tilde\xi_2))+ h^2\mathcal{J}_{h}\Op_{ h}^w\nu^2(x_2, h\tilde\xi_2) +h^2 g^\star(h, \mathcal{J}_{ h},x_{2},h\tilde\xi_{2})\\
+\mathsf{remainders},
\end{multline*}
where $\mathcal{J}_{h}=\Op_{h}^w\left(\tilde\xi_{3}^2+x_{3}^2\right)$
and $g^\star (h, J, x_{2}, \xi_{2})$ is of order three
with respect to $(J^{\frac{1}{2}}, h^{\frac{1}{2}})$. Motivated again by the
perspective of describing the low lying eigenvalues, we replace
$\mathcal{J}_{h}$ by $h$ and rewrite the symbol with the old
semiclassical parameter $\hbar$ to get the operator
$\mathcal{M}^{[1]}_{\hbar}=\Op_{h}^w\left(\mathsf{M}_{h}^{[1]}\right)=\Op_{\h}^w\left(M_{\h}^{[1]}\right)$,
with
\begin{equation}\label{nf3}
M_{\hbar}^{[1]}=\hbar b(x_2, \xi_2,s(x_2, \xi_2))+ \hbar^{\frac{3}{2}}\nu^2(x_2,\xi_2) +\hbar g^{\star}(\hbar^{\frac{1}{2}},\hbar^{\frac{1}{2}},x_{2},\xi_{2})+\mathsf{remainders}.
\end{equation}

\subsubsection{Third Birkhoff form}
The last generic assumption is the uniqueness and non-degeneracy of
the minimum of the new \enquote{principal} symbol
$$(x_{2},\xi_{2})\mapsto b(x_2, \xi_2,s(x_2, \xi_2))$$ 
that implies that $b$ admits a unique and non-degenerate minimum at
$(0,0,0)$. Up to an $\hbar^{\frac{1}{2}}$-dependent translation in the
phase space and a rotation, we are essentially reduced to a standard
Birkhoff normal form with respect to the third harmonic oscillator
$\mathcal{K}_{\hbar}=\hbar^2D_{x_{2}}^2+x_{2}^2$.

Note that all our normal forms may be used to describe the classical dynamics of a charged particle in a confining magnetic field (see Figure \ref{bouteille}).

\begin{figure}[ht]
  \includegraphics[width=0.75\textwidth]{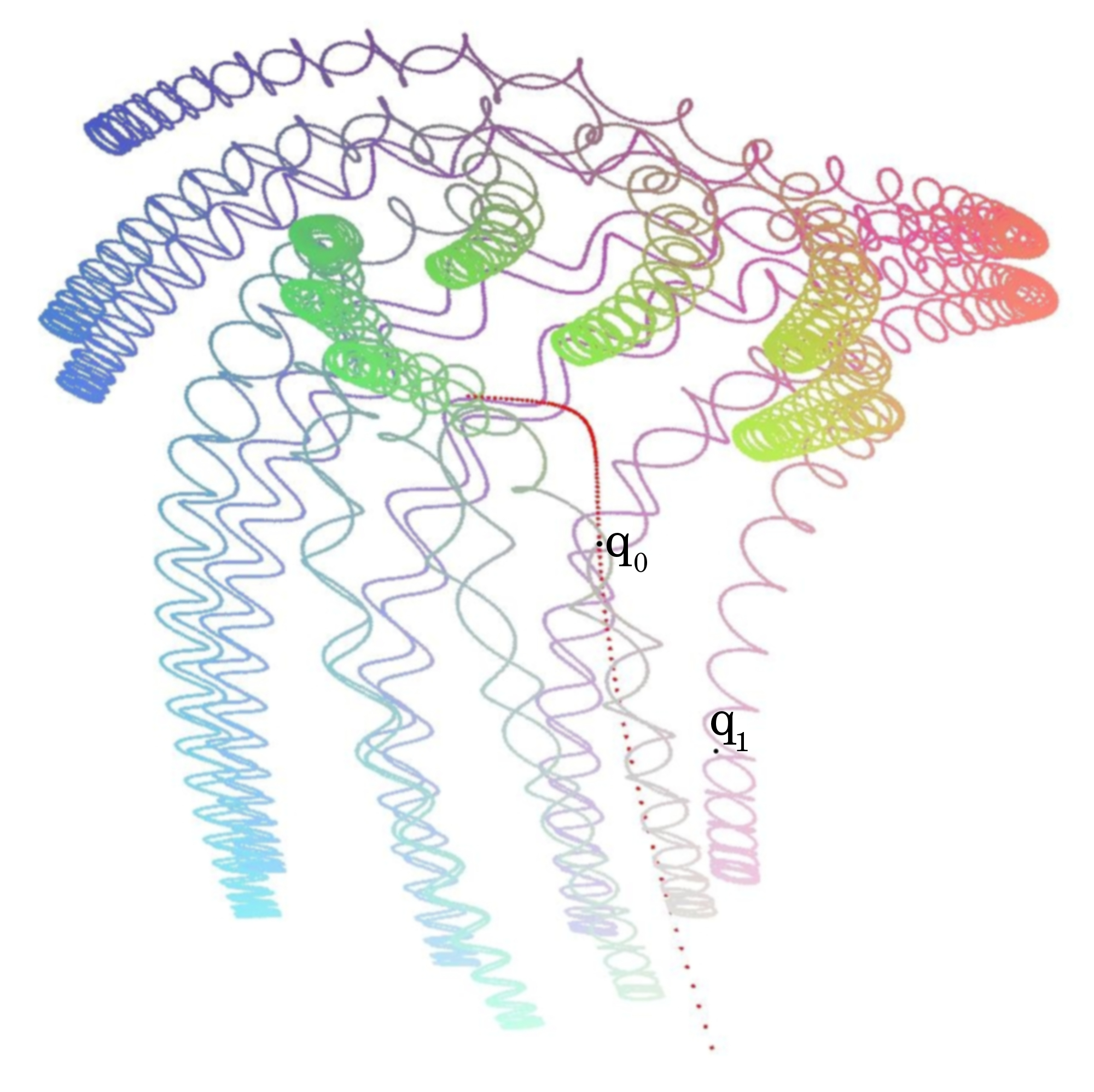}
  \caption{ The dashed line represents the integral curve of the
    confining magnetic field $\bB ={ \rm curl} {\bf A}$ through
    $q_{0}=(0.5, 0.6, 0.7 )$ for
    $\bB(x,y,z)=\left(\frac{y}{2},\frac{z}{2},\sqrt{1+x^2}\right)$ and
    the full line represents the projection in the $q$-space of the
    Hamiltonian trajectory with initial condition $(q_0, p_{0})$ (with
    $p_{0}= (-0.6, 0.01, 0.2)$) ending at
    $(q_{1},p_{1})$.}\label{bouteille}
\end{figure}

\subsubsection{Microlocalization}
Of course, at each step, we will have to provide accurate microlocal
estimates of the eigenfunctions of the different operators to get a
good control of the different remainders. In a first approximation, we
will get localizations at the following scales $x_{1}, \xi_{1},
\xi_{3}\sim \hbar^{\delta}$ ($\delta>0$ is small enough) and $x_{2},
\xi_{2}, x_{3}\sim 1$. In a second approximation, we will get $x_{3},
\tilde \xi_{3}\sim \hbar^{\delta}$. In the final step, we will refine the localization by $x_{2}, \xi_{2}\sim \hbar^{\delta}$.

\subsection{A semiclassical eigenvalue estimate} 
Let us already state one of the consequences of our investigation. It
will follow from the third normal form that we have a complete
description of the spectrum below the threshold $b_{0}\hbar+3
\nu^2(0,0)\hbar^{\frac{3}{2}}$. This description is reminiscent of the
results \textit{\`a la} Bohr-Sommerfeld of \cite{HelRo84} and
\cite[Appendix B]{HelSj89} (see also \cite[Remark 1.4]{HelKo13b})
obtained in the case of one dimensional semiclassical operators.

\begin{theo}
  \label{t:main}
  Assume that $b$ admits a unique and non degenerate minimum at
  $q_{0}$.  Denote
   \begin{equation}  \sigma=\frac{\mathsf{Hess}_{q_0} b \,(\bB,
      \bB)}{2b_0^2},\quad \dd=\sqrt{\frac{\det {\mathsf{Hess}}_{q_0}
        b}{\mathsf{Hess}_{q_0} b \,(\bB, \bB)}}. \label{equ:alpha} 
    \end{equation}
     There exists a function
  $k^\star\in\mathcal{C}^\infty_0(\R^2)$ with arbitrarily small
  compact support, and $k^\star(\hbar^{\frac{1}{2}},Z)=\mathcal{O}((\h
  + \abs{Z})^{\frac{3}{2}})$ when $(\h,Z)\to(0,0)$, such that the
  following holds.

  For all $c\in(0,3)$, the spectrum of $\mathcal{L}_{\hbar,\bA}$ below
  $b_{0}\hbar+c\sigma^{\frac12}\hbar^{\frac{3}{2}}$ coincides modulo
  $\OO(\hbar^\infty)$ with the spectrum of the operator
  $\mathcal{F}_{\hbar}$ acting on $L^2(\R_{x})$ given by
  \[
  \mathcal{F}_{\hbar}=b_{0}\hbar +
  \sigma^{\frac12}\hbar^{\frac{3}{2}}- \frac{\zeta}{2\dd}\hbar^2+
  \hbar\left(\frac{\dd}2 \mathcal{K}_{\hbar}
    +k^\star(\hbar^{\frac{1}{2}}, \mathcal{K}_{\hbar})\right) , \quad
  \ \mathcal{K}_\h=\hbar^2D_{x}^2+x^2 ,
  \]
  with some constant $\zeta$.
\end{theo}

\begin{rem}
The constant $\zeta$ in Theorem \ref{t:main} is given by the formula
 \[
    \zeta=\|\nabla\nu^2(0,0)\|^2, 
 \] 
where the function $\nu$ is given in \eqref{eq.nu}. Observe also that $\sigma=\nu^4(0,0)$.
\end{rem}

  \begin{corollary}\label{coro:main} Under the hypothesis
    of Theorem~\ref{t:main}, let $\left(\lambda_m(\h)\right)_{m\geq
      1}$ be the non decreasing sequence of the eigenvalues of
    $\mathcal L_{\h,\bA}$.  For any $c\in (0,3)$, let \[
    \mathsf{N}_{\h,c} := \{m\in \N^*; \quad \lambda_m(\h) \leq \hbar
    b_0 +c\sigma^{\frac{1}{2}}\hbar^{\frac{3}{2}}\}.  \] Then the cardinal of
    $\mathsf{N}_{\h,c}$ is of order $\h^{-\frac{1}{2}}$, and there exist
    $\upsilon_1, \upsilon_2\in\R$ and $\h_0>0$ such
    that \begin{equation*}\label{appli} \lambda_m(\h) \!=\! \hbar b_0
      +\sigma^{\frac12}\hbar^{\frac32}+ \left[{\dd}(m-\tfrac12) -
        \frac{\zeta}{2\dd} \right]\!\hbar^2 +
      \upsilon_1(m-\tfrac12)\h^{\frac52} +
      \upsilon_2(m-\tfrac12)^2\h^3 + \OO(
      \hbar^{\frac{5}{2}})\,, \end{equation*} uniformly for $\h\in(0,\h_0)$ and $m\in \mathsf{N}_{\h,c}$.  
      
 In particular, the splitting between two consecutive eigenvalues satisfies
 $$\lambda_{m+1}(\h)-\lambda_{m}(\h) = \dd\hbar^2 + \OO(\h^{\frac52}).$$         
  \end{corollary}

\begin{proof}
  If the support of $k^\star$ is small enough, the hypothesis
  $k^\star(\hbar^{\frac{1}{2}},Z)=\mathcal{O}((\h +
  \abs{Z})^{\frac{3}{2}})$ implies that, when $\h$ is small enough,
\[
(1+\eta)  \mathcal{K}_{\hbar} \geq  \mathcal{K}_{\hbar}
    + \frac2{\dd} k^\star(\hbar^{\frac{1}{2}}, \mathcal{K}_{\hbar}) \geq (1-\eta)  \mathcal{K}_{\hbar},
\]
for some small $\eta>0$. Therefore, since the eigenvalues of
$\mathcal{K}_{\hbar}$ are $(2m-1)\h$, $m\in\N^*$, the variational
principle implies that the number of eigenvalues of
$\mathcal{K}_{\hbar} +\frac2{\dd}k^\star(\hbar^{\frac{1}{2}},
\mathcal{K}_{\hbar})$ below a threshold $C_\h$ belongs to
$[\frac{1}{2}(\frac{C_\h}{\h(1+\eta)}+1),
\frac{1}{2}(\frac{C_\h}{\h(1-\eta)}+1)]$. Taking
$C_\h=\frac{2}{\dd}(c-1)\sigma^{1/2}\h^{1/2} + \frac{\zeta}{\dd^2}\h$,
and applying the theorem, we obtain the estimate for the cardinal of
$\mathsf{N}_{\h,c}$. The corresponding eigenvalues of $\mathcal
L_{\h,\bA}$ are of the form
\[
\lambda_m(\h) = \hbar b_0
      +\sigma^{\frac12}\hbar^{\frac32}
      - \frac{\zeta}{2\dd}\hbar^2
      + \h\left[{\dd}(m-\tfrac12) 
        +k^\star(\hbar^{\frac{1}{2}}, 2m-1)\right] + \OO(\h^\infty) ,
\]
with $(2m-1)\h\leq \frac{C_\h}{1-\eta}$. Therefore there exists a
constant $\tilde C>0$, independent of $\h$, such that all $m\in
\mathsf{N}_{\h,c}$ satisfy the inequality $(2m-1)\h \leq \tilde
C\h^{1/2}$. Writing
\begin{multline*}
k^\star(\hbar^{\frac{1}{2}},Z) = c_0\h^{3/2} +\upsilon_1\h^{1/2} (Z/2) +
c_1\h^2 + \upsilon_2 (Z/2)^2 + \upsilon_3 \h Z \\
+ \h^{1/2}\OO(h+\abs{Z})^2 + \OO(Z^3) ,
\end{multline*}
we see that, for $m\in \mathsf{N}_{\h,c}$,
\[
k^\star(\hbar^{\frac{1}{2}},(2m-1)\h) = \upsilon_1\h^{3/2}\left(m-\frac{1}{2}\right)  +
 \upsilon_2 \h^2\left(m-\frac{1}{2}\right)^2 + \OO(\h^{3/2}) ,
\]
which gives the result.
\end{proof}
\begin{rem}
 An upper bound of $\lambda_{m}(\h)$ for fixed $\h$-independent $m$ with remainder in
  $\mathcal{O}(\h^\frac{9}{4})$ was obtained in \cite{HelKo13} through
  a quasimodes construction involving powers of $\hbar^{\frac{1}{4}}$. To the authors' knowledge, Corollary \ref{coro:main} gives the most accurate description
  of magnetic eigenvalues in three dimensions, in such a large
  spectral window. Note also that the non-degeneracy
  assumption on the norm of $\bB$ is not purely technical. Indeed, at the quantum level, it appears through microlocal reductions matching with the splitting of the Hamiltonian dynamics into three
  scales: the cyclotron motion around field lines, the center-guide oscillation along the field lines, and the oscillation
  within the space of field lines. 
 \end{rem}

\subsection{Organization of the paper}
The paper is organized as follows. In Section \ref{results}, we state
our main results. Section \ref{BNF1} is devoted to the investigation
of the first normal form (see Theorem \ref{normal-form1} and Corollary
\ref{normal-form1-d}). In Section \ref{BNF2} we analyze the second
normal form (see Theorems \ref{pre-normal-form2} and
\ref{normal-form2} and Corollaries \ref{pre-normal-form2-d} and
\ref{normal-form2-d}). Section \ref{BNF3} is devoted to the third
normal form (see Theorem \ref{normal-form3} and Corollary
\ref{normal-form3-d}).

\section{Statements of the main results}\label{results}

We recall (see \cite[Chapter 7]{DiSj99}) that a function $m:\R^d\to
[0,\infty)$ is an order function if there exist constants $N_0, C_0>0$
such that 
\[
m(X)\leq C_0\langle X-Y\rangle^{N_0} m(Y)
\]
for any $X,Y\in\R^d$. The symbol class $S(m)$ is the space of smooth
$\h$-dependent functions $a_\h:\R^d\to\C$ such that
\[
\forall
\alpha\in\N^d, \qquad \abs{\partial_x^\alpha a_\h(x)} \leq C_\alpha m(x), \qquad  \forall h\in(0,1].
\]

Throughout this paper, we assume that the components of the vector
potential $\bA$ belong to a symbol class $S(m)$. Note that this
implies that $\bB\in S(m)$, and conversely, if $\bB\in S(m)$, then
there exist a potential $\bA$ and another order function $m'$ such
that $\bA\in S(m')$. Moreover, the magnetic Hamiltonian
$H(x,\xi)=\norm{\xi-\bA(x)}^2$ belongs to $S(m'')$ for an order
function $m''$ on $\R^6$.

We will work with the Weyl quantization; for a classical symbol
$a_{\h}=a(x,\xi;\h)\in S(m)$, it is defined as:
\begin{equation}\label{Weyl}
\Op_{\h}^w a\, \psi(x)=\frac{1}{(2\pi\h)^d}\int_{\R^{2d}} e^{i\pscal{x-y}{\xi}/\h} a\left(\frac{x+y}{2}, \xi\right)\psi(y)\dx y\dx \xi,\quad \forall \psi\in\mathcal{S}(\R^d).
\end{equation}
The Weyl quantization of $H$ is the magnetic Laplacian
$\mathcal{L}_{\h,\bA}=(-i\h\nabla-\bA)^2$.

\subsection{Normal forms and spectral reductions}
Let us introduce our first Birkhoff normal form $\mathcal{N}_{\hbar}$.
\begin{theo}\label{normal-form1}
  If $\bB(0)\neq 0$, there exists a neighborhood of $(0,\bA(0))$
  endowed with symplectic coordinates $(x_{1},\xi_{1}, x_{2},\xi_{2},
  x_{3},\xi_{3})$ in which $\Sigma=\{x_{1}=\xi_{1}=\xi_{3}=0\}$ and $(0,\bA(0))$ has coordinates $0\in\R^6$, and
  there exist an associated unitary Fourier integral operator
  $U_{\hbar}$ and a smooth function, compactly supported with respect
  to $Z$ and $\xi_{3}$, $f^\star(\hbar,Z, x_{2}, \xi_{2}, x_{3},
  \xi_{3})$ whose Taylor series with respect to $Z, \xi_{3},\hbar$ is
\[
\sum_{k\geq 3}\sum_{2\ell+2m+\beta=k}\hbar^\ell
c^\star_{\ell,m,\beta}(x_2,\xi_2,x_3) Z^{m}\xi_3^\beta
\]
such that
\begin{equation}\label{Nh}
U_{\hbar}^*\mathcal{L}_{\hbar,\bA}U_{\hbar}=\mathcal{N}_{\hbar}+\mathcal{R}_{\hbar},
\end{equation}
with $$\mathcal{N}_{\hbar}=\hbar^2D_{x_{3}}^2+\mathcal{I}_{\hbar}\Op_{\hbar}^w b+ \Op_{\hbar}^w f^\star(\hbar,\mathcal{I}_{\hbar}, x_{2}, \xi_{2}, x_{3}, \xi_{3}),$$
and where 
\begin{enumerate}[(a)]
\item we have $\mathcal{I}_{\hbar}=\hbar^2D_{x_{1}}^2+x_{1}^2$,
\item the operator $\Op_{\hbar}^w f^\star(\hbar,\mathcal{I}_{\hbar}, x_{2}, \xi_{2}, x_{3}, \xi_{3})$ has to be understood as the Weyl quantization of an operator valued symbol,
\item the remainder $\mathcal{R}_{\hbar}$ is a pseudo-differential
  operator such that, in a neighborhood of the origin, the Taylor
  series of its symbol with respect to $(x_{1}, \xi_{1}, \xi_{3}, \hbar)$ is $0$.
\end{enumerate}
\end{theo}
\begin{rem}\label{rem.thm1}
In Theorem \ref{normal-form1}, the direction of $\bB$ considered as a vector field on $\Sigma$ is $\frac{\partial}{\partial x_{3}}$ and the function $b\in \mathcal{C}^\infty(\mathbb R^6)$ stands for $b\circ j^{-1}_{|\Sigma}\circ\pi$ where $\pi : \mathbb R^6 \to \Sigma : \pi(x_{1},\xi_{1}, x_{2},\xi_{2},x_{3},\xi_{3})=(0,0, x_{2},\xi_{2},x_{3},0)$. In addition, note that the support of $f^{\star}$ in $Z$ and $\xi_{3}$ may be chosen as small as we want.
\end{rem}
\begin{rem}
  In the context of Weyl's asymptotics, a close version of this
  theorem appears in \cite[Chapter 6]{I98}.
\end{rem}

In order to investigate the spectrum of $\mathcal{L}_{\hbar,\bA}$ near the low lying energies, we introduce the following pseudo-differential operator
$$\mathcal{N}^{[1]}_{\hbar}=\hbar^2D_{x_{3}}^2+\hbar\Op_{\hbar}^w b+ \Op_{\hbar}^w f^\star(\hbar,\hbar, x_{2}, \xi_{2}, x_{3}, \xi_{3}),$$
obtained by replacing $\mathcal{I}_{\hbar}$ by $\hbar$. 

\begin{corollary}\label{normal-form1-d}
We introduce
\begin{equation}\label{Nheta}
\mathcal{N}^\sharp_{\hbar}=\Op_{\hbar}^w\left(N_{\hbar}^\sharp\right),
\end{equation}
with
$$N_{\hbar}^\sharp=\xi_{3}^2+\mathcal{I}_{\hbar}\underline{b}(x_{2},\xi_{2},x_{3})+ f^{\star,\sharp}(\hbar,\mathcal{I}_{\hbar}, x_{2}, \xi_{2}, x_{3}, \xi_{3})$$
and where $\underline{b}$ is a smooth extension of $b$ away from $D(0,\eps)$ such that \eqref{hyp4} still holds and where  $f^{\star,\sharp}=\chi(x_{2},\xi_{2}, x_{3})f^{\star}$, with $\chi$ is a smooth cutoff function being $1$ in a neighborhood of $D(0,\eps)$. We also define the operator attached to the first eigenvalue of $\mathcal{I}_{\hbar}$
\begin{equation}\label{Nheta1}
\mathcal{N}^{[1],\sharp}_{\hbar}=\Op_{\hbar}^w\left(N_{\hbar}^{[1],\sharp}\right),
\end{equation}
where $N_{\hbar}^{[1],\sharp}=\xi_{3}^2+\hbar\underline{b}(x_{2},\xi_{2},x_{3})+ f^{\star,\sharp}(\hbar,\hbar, x_{2}, \xi_{2}, x_{3}, \xi_{3})$.

If $\eps$ and the support of $f^\star$ are small enough, then we have
\begin{enumerate}[(a)]
\item \label{normal-form1-d-a}The spectra of $\mathcal{L}_{\hbar,\bA}$ and $\mathcal{N}^\sharp_{\hbar}$ below $\beta_{0}\hbar$ coincide modulo $\OO(\hbar^\infty)$.
\item \label{normal-form1-d-b} For all $c\in(0,\min(3b_{0},\beta_{0}))$, the spectra of $\mathcal{L}_{\hbar,\bA}$ and $\mathcal{N}^{[1],\sharp}_{\hbar}$ below $c\hbar$ coincide modulo $\OO(\hbar^\infty)$.
\end{enumerate}

\end{corollary}
Let us now state our results concerning the normal form of $\mathcal{N}_{\hbar}^{[1]}$ (or $\mathcal{N}_{\hbar}^{[1],\sharp}$) under the following assumption. 
\begin{notation}
If $f=f(\z)$ is a differentiable function, we denote by $T_{\z}f(\cdot)$ its tangent map at the point $\z$. Moreover, if $f$ is twice differentiable, the second derivative of $f$ is denoted by $T^2_{\z}f(\cdot,\cdot)$.
\end{notation}
\begin{hyp}\label{non-deg}
We assume that $T^2_{0}b(\bB(0),\bB(0))>0$.
\end{hyp}

\begin{rem}
If the function $b$ admits a unique and positive minimum at $0$ and that it is non degenerate, then Assumption \ref{non-deg} is satisfied.
\end{rem}

Under Assumption \ref{non-deg}, we have $\partial_3b(0,0,0)=0$ and, in the coordinates $(x_{2},\xi_{2},x_{3})$ given in Theorem \ref{normal-form1},
\begin{equation}
\partial^2_3b(0,0,0) > 0\,.\label{eq:positive}
\end{equation}
 It follows
from~\eqref{eq:positive} and the implicit function theorem that, for
small $x_2$, there exists a smooth function $(x_2,\xi_2)\mapsto s(x_2,\xi_2)$,
$s(0,0)=0$, such that
\begin{equation}
\partial_3b(x_2,\xi_2,s(x_2,\xi_2)) = 0\,.\label{eq:critical}
\end{equation}
The point $s(x_{2},\xi_{2})$ is the unique (in a neighborhood of $(0,0,0)$) minimum of $x_{3}\mapsto b(x_{2},\xi_{2},x_{3})$.
We define
$$\nu(x_{2},\xi_{2}):= (\tfrac{1}{2}\partial_3^2 b(x_2,\xi_2,s(x_2,\xi_2)))^{1/4}.$$
\begin{theo}\label{pre-normal-form2}
Under Assumption \ref{non-deg}, there exists a neighborhood $\mathcal{V}_{0}$ of $0$ and a Fourier integral operator $V_{\hbar}$ which is microlocally unitary near $\mathcal{V}_{0}$ and such that
$$V_{\hbar}^*\mathcal{N}_{\hbar}^{[1]}V_{\hbar}=:\underline{\mathcal{N}}^{[1]}_{\hbar}=\Op_{\hbar}^w\left(\underline{N}^{[1]}_{\hbar}\right),$$
where $\underline{N}^{[1]}_{\hbar}=\nu^2(x_2,\xi_2)\left(\xi_3^2 + \h
  x_3^2\right)+ \h b(x_2,\xi_2,s(x_2,\xi_2))+\underline{r}_{\hbar}$
and $\underline{r}_{\hbar}$ is a semiclassical symbol such that
$\underline{r}_{\hbar}=\OO(\h x_3^3) +\OO(\h\xi_3^2) + \OO(\xi_3^3) +
\OO(\h^2)$.
\end{theo}

\begin{corollary}\label{pre-normal-form2-d}
Let us introduce
$$\underline{\mathcal{N}}^{[1],\sharp}_{\hbar}=\Op_{\hbar}^w\left(\underline{N}^{[1],\sharp}_{\hbar}\right),$$
where $\underline{N}^{[1],\sharp}_{\hbar}=
\underline{\nu}^2(x_2,\xi_2)\left(\xi_3^2 + \h x_3^2\right)+ \h
\underline{b}(x_2,\xi_2,s(x_2,\xi_2))+\underline{r}^\sharp_{\hbar},$
with
$\underline{r}^\sharp_{\hbar}=\chi(x_{2},\xi_{2},x_{3},\xi_{3})\underline{r}_{\hbar}$,
and where $\underline{\nu}$ denotes a smooth and constant (with a
positive constant) extension of the function $\nu$.

There exists a constant $\tilde c>0$ such that, for any cut-off
function $\chi$ equal to 1 on $D(0,\eps)$ with support in
$D(0,2\eps)$, we have:
\begin{enumerate}[(a)]
\item\label{pre-normal-form2-d-a} The spectra of $\underline{\mathcal{N}}^{[1],\sharp}_{\hbar}$ and $\mathcal{N}^{[1],\sharp}_{\hbar}$ below $(b_{0}+\tilde c\eps^2)\hbar$ coincide modulo $\OO(\hbar^{\infty})$.
\item\label{pre-normal-form2-d-b} For all $c\in(0,\min(3b_{0},b_0+\tilde c\eps^2))$, the spectra of $\mathcal{L}_{\hbar,\bA}$ and $\underline{\mathcal{N}}^{[1],\sharp}_{\hbar}$ below $c\hbar$ coincide modulo $\OO(\hbar^{\infty})$.
\end{enumerate}
\end{corollary}

\begin{notation}[Change of semiclassical parameter]\label{notation}
  We let $h=\hbar^{\frac{1}{2}}$ and, if $A_{\hbar}$ is a
  semiclassical symbol on $T^*\R^2$, admitting a semiclassical
  expansion in $\hbar^{\frac{1}{2}}$, we write
\[
\mathcal{A}_{\hbar}:=\Op_{\hbar}^w A_{\hbar}=\Op_{h}^w \mathsf{A}_{ h}=:\mathfrak{A}_{h},
\]
with
\[
\mathsf{A}_{h}(x_{2},\tilde \xi_{2}, x_{3}, \tilde \xi_{3}) = 
A_{h^2}(x_{2}, h\tilde\xi_{2}, x_{3}, h\tilde\xi_{3}).
\]
Thus, $\mathcal{A}_{\hbar}$ and $\mathfrak{A}_{h}$ represent the same
operator when $h=\h^{\frac12}$, but the former is viewed as an
$\h$-quantization of the symbol $A_{\hbar}$, while the latter is an
$h$-pseudo-differential operator with symbol $\mathsf{A}_{ h}$. Notice
that, if $A_{\hbar}$ belongs to some class $S(m)$, then $\mathsf{A}_{
  h}\in S(m)$ as well. This is of course not true the other way
around.
\end{notation}

\begin{theo}\label{normal-form2}
  Under Assumption \ref{non-deg}, there exist a unitary operator
  $W_{h}$ and a smooth function $g^\star(h, Z, x_{2},\xi_{2})$, with compact support as
  small as we want with respect to $Z$ and with
  compact support in $(x_{2},\xi_2)$, whose Taylor series with
  respect to $Z$, $h$ is
$$\sum_{2m+2\ell\geq 3} c_{m,\ell}(x_{2}, \xi_{2}) Z^{m}  h^\ell,$$
such that
$$W^*_{h}\underline{\mathfrak{N}}^{[1],\sharp}_{h}W_{h}=:\mathfrak{M}_{h}=\Op_{h}^w\left(\mathsf{M}_{h}\right),$$
with
\begin{multline*}
  \mathsf{M}_{h}=h^2\underline{b}(x_2,h \tilde\xi_2,s(x_2,
  h\tilde\xi_2))+ h^2\mathcal{J}_{h}\Op_{ h}^w\underline{\nu}^2(x_2,
  h\tilde\xi_2) +h^2 g^\star(h, \mathcal{J}_{
    h},x_{2},h\tilde\xi_{2})\\
+h^2\mathsf{R}_{h} + h^\infty S(1).
\end{multline*}

where 
\begin{enumerate}[(a)]
\item the operator $\underline{\mathfrak{N}}^{[1],\sharp}_{h}$ is
  $\underline{\mathcal{N}}^{[1],\sharp}_{\hbar}$ (but written in the
  $h$-quantization),
\item we have let $\mathcal{J}_{ h}=\Op_{h}^w\left(\tilde\xi_{3}^2+x_{3}^2\right)$,
\item the function $\mathsf{R}_{h}$ satisfies
  $\mathsf{R}_{h}(x_{2},h\tilde \xi_{2}, x_{3},
  \tilde\xi_{3})=\OO((x_{3},\tilde\xi_{3})^\infty)$.
\end{enumerate}
\end{theo}

\begin{rem}
Note that the support of $g^\star$ with respect to $Z$ may be chosen as small as we want. Note also that we have used $\underline{\mathfrak{N}}^{[1],\sharp}_{h}$ instead of $\underline{\mathfrak{N}}^{[1]}_{h}$: Since $W_{h}$ is exactly unitary, we get a direct comparison of the spectra.
\end{rem}

\begin{corollary}\label{normal-form2-d}
We introduce
$$\mathfrak{M}^\sharp_{h}=\Op_{h}^w\left(\mathsf{M}_{h}^{\sharp}\right),$$
with
$$\mathsf{M}_{h}^{\sharp}=h^2 \underline{b}(x_2,h \tilde\xi_2,s(x_2, h\tilde\xi_2))+h^2\mathcal{J}_{h}\underline{\nu}^2(x_2,h\tilde\xi_2) + h^2 g^{\star}(h, \mathcal{J}_{ h},x_{2},h\tilde\xi_{2}).$$
We also define
$$\mathfrak{M}^{[1],\sharp}_{h}=\Op_{h}^w\left(\mathsf{M}_{h}^{[1],\sharp}\right),$$
with 
$$\mathsf{M}_{h}^{[1],\sharp}=h^2 \underline{b}(x_2, h \tilde\xi_2,s(x_2, h\tilde\xi_2))+ h^3\underline{\nu}^2(x_2,h\tilde\xi_2) + h^2 g^{\star}(h,h,x_{2},h\tilde\xi_{2}).$$
If $\eps$ and the support of $g^{\star}$ are small enough, we have
\begin{enumerate}[(a)]
\item\label{normal-form2-d-a} For all $\eta>0$, the spectra of $\underline{\mathfrak{N}}^{[1],\sharp}_{h}$ and $\mathfrak{M}^\sharp_{h}$ below $b_{0}h^2+\OO(h^{2+\eta})$ coincide modulo $\OO(h^\infty)$.
\item\label{normal-form2-d-b} For $c\in(0,3)$, the spectra of $\mathfrak{M}^{\sharp}_{h}$ and $\mathfrak{M}^{[1],\sharp}_{h}$ below $b_{0}h^2+c\sigma^{\frac{1}{2}}h^3$ coincide modulo $\OO(h^\infty)$.
\item\label{normal-form2-d-c} If $c\in(0,3)$, the spectra of $\mathcal{L}_{\hbar,\bA}$ and $\mathcal{M}^{[1],\sharp}_{\hbar}=\mathfrak{M}^{[1],\sharp}_{h}$ below $b_{0}\hbar+c\sigma^{\frac{1}{2}}\hbar^{\frac{3}{2}}$ coincide modulo $\OO(\hbar^\infty)$.
\end{enumerate}
\end{corollary}
Finally, we can perform a last Birkhoff normal form for the operator $\mathcal{M}^{[1],\sharp}_{\hbar}$ as soon as $(x_{2},\xi_{2})\mapsto\underline{b}(x_{2},\xi_{2},s(x_{2},\xi_{2}))$ admits a unique and non degenerate minimum at $(0,0)$. Under this additional assumption, $b$ admits a unique and non degenerate minimum at $(0,0,0)$.

Therefore we will use the following stronger assumption.
\begin{hyp}\label{non-deg2}
The function $b$ admits a unique and positive minimum at $0$ and it is non degenerate.
\end{hyp}

\begin{theo}\label{normal-form3}
Under Assumption \ref{non-deg2}, there exist a unitary $\hbar$-Fourier Integral Operator $Q_{\hbar^{\frac{1}{2}}}$ whose phase admits an expansion in powers of $\hbar^{\frac{1}{2}}$ such that
$$Q^*_{\hbar^{\frac{1}{2}}}\mathcal{M}^{[1],\sharp}_{\hbar}Q_{\hbar^{\frac{1}{2}}}=\mathcal{F}_{\hbar}+\mathcal{G}_{\hbar},$$
where
\begin{enumerate}[(a)]
\item $\mathcal{F}_{\hbar}$ is defined in Theorem \ref{t:main},
\item the remainder is in the form $\mathcal{G}_{\hbar}=\Op_{\hbar}^w
  \left(G_{\hbar}\right)$, with $G_{\hbar}=\hbar\OO(|z_{2}|^\infty)$.
\end{enumerate}
\end{theo}

\begin{corollary}\label{normal-form3-d}
If $\eps$ and the support of $k^\star$ are small enough, we have
\begin{enumerate}[(a)]
\item\label{normal-form3-d-a} For all $\eta\in\left(0,\frac{1}{2}\right)$, the spectra of $\mathcal{M}^{[1],\sharp}_{\hbar}$ and $\mathcal{F}_{\hbar}$ below $b_{0}\hbar+\OO(\hbar^{1+\eta})$ coincide modulo $\OO(\hbar^\infty)$.
\item\label{normal-form3-d-b} For all $c\in(0,3)$, the spectra of $\mathcal{L}_{\hbar,\bA}$ and $\mathcal{F}_{\hbar}$ below $b_{0}\hbar+c\sigma^{\frac{1}{2}}\hbar^{\frac{3}{2}}$ coincide modulo $\OO(\hbar^\infty)$.
\end{enumerate}
\end{corollary}

\begin{rem}
Since the spectral analysis of $\mathcal{F}_{\hbar}$ is straightforward, Item \eqref{normal-form3-d-b} of Corollary \ref{normal-form3-d} implies Theorem \ref{t:main}.
\end{rem}
The next sections are devoted to the proofs of our main results.

\noindent \resizebox{\textwidth}{!}{\input{circ}}

\section{First Birkhoff normal form}\label{BNF1}
We assume that $\bB(0)\neq 0$ so that in some neighborhood $\Omega$ of $0$ the magnetic field does not vanish. Up to a
rotation in $\R^3$ (extended to a symplectic transformation in $\R^6$)
we may assume that $\bB(0)=\|\bB(0)\|\e_{3}$. In this neighborhood, we may defined the unit vector:
\begin{equation}\label{b}
\bb=\frac{\bB}{\|\bB\|}
\end{equation}
and find vectors $\bc$ and $\bd$ depending smoothly on $q$ such that $(\bb, \bc, \bd)$ is a direct orthonormal basis.

\subsection{Symplectic coordinates}

\subsubsection{Straightening the magnetic vector field}
We consider the form $\dx \alpha$ and we would like to find a diffeomorphism, in a neighborhood of $0$, $\chi$ such that $\chi(\hat q)=q$ and $\chi^*(\dx\alpha)=\dx \hat q_{1}\wedge \dx\hat q_{2}$. First, this is easy to find a local diffeomorphism $\varphi$ such that
$$\dr_{3}\varphi(\tilde q)=\bb(\varphi(\tilde q))$$
and $\varphi(\tilde q_{1},\tilde q_{2},0)=(\tilde q_{1},\tilde
q_{2},0)$. This is just the standard straigthening-out lemma for the
non-vanishing vector field $\bb$.

The vector $\e_3$ is in the kernel of $\varphi^*(\dx\alpha)$, which
implies that we have $\varphi^*(\dx\alpha) = f(\tilde{q})\dx \tilde{q}_1\wedge \dx
\tilde{q}_2$, for some smooth function $f$.

But since the form $\varphi^*(\dx\alpha)$ is closed, $f$ does not depend on $\tilde q_{3}$. This is then easy to find another diffeomorphism $\psi$, corresponding to the change of variables 
$$\hat q=\psi(\tilde q)=(\psi_{1}(\tilde q_{1},\tilde q_{2}),\psi_{2}(\tilde q_{1},\tilde q_{2}),\tilde q_{3})\, ,$$ 
such that
$$\psi^*(\varphi^*(\dx\alpha))=\dx\hat q_{1}\wedge \dx\hat q_{2}\,.$$
We let $\chi=\varphi\circ\psi$ and we notice that
\begin{equation}\label{chi*}
\chi^*(\dx \alpha)=\dx\hat q_{1}\wedge \dx\hat q_{2}\qquad\dr_{3}\chi(\hat q)=\bb(\chi(\hat q))\,,
\end{equation}
\begin{rem}\label{rem.Tchi}
It follows from \eqref{chi*} and \eqref{MB} that $\det T\chi=\|\bB\|^{-1}$.
\end{rem}
\subsubsection{Symplectic coordinates}
Let us consider the new parametrization of $\Sigma$ given by
\[
\begin{aligned}
  \iota: \hat{\Omega} & \longrightarrow \Sigma\\
\hat{q} & \mapsto (\chi(\hat{q}), A_1(\chi(\hat{q}))\,,
A_2(\chi(\hat{q})),A_3(\chi(\hat{q})) )\,,
\end{aligned}
\]
which gives a basis $(\f_1,\f_{2},\f_3)$ of $T\Sigma$~:
\[
\f_j= (T\chi(\e_j), T\bA\circ T\chi(\e_j)), \; j=1,2,3\,.
\]
Using \eqref{lemma1}, and the fact that $\f_3$ is in the
kernel of $\dx \alpha$, we find $\omega_0(\f_j,\f_3)=0$, $j=1,2\,$.
Finally, $\omega_0(f_1,f_2) = \dx \alpha(T\chi \e_1, T\chi \e_2) =
\chi^*(\dx \alpha)(\e_1,\e_2) = 1\,$.

The following vectors of $\R^3\times\R^3$ form a basis of the symplectic orthogonal of $T_{\iota(\hat q)}\Sigma$:
\begin{equation}\label{f45}
\f_{4}=\|\bB\|^{-1/2}(\bc, ({}^t T_{\chi(\hat q)}\bA)\bc),\quad\f_{5}=\|\bB\|^{-1/2}(\bd, ({}^t T_{\chi(\hat q)}\bA)\bd),
\end{equation}
so that
$$\omega_{0}(\f_{4},\f_{5})=-1.$$
We let $\f_{6}=(0,\bb)+\rho_{1} \f_{1}+\rho_{2} \f_{2}$ where $\rho_{1}$ and $\rho_{2}$ are determined so that $\omega_{0}(\f_{j},\f_6)=0$ for $j=1,2$.
We notice that $\omega_{0}(\f_{j},\f_{6})=0$ for $j=4,5$ and $\omega_{0}(\f_{3},\f_{6})=-1$.

\subsubsection{Diagonalizing the Hessian}
We recall that
$$H(q,p)=\|p-\bA(q)\|^2$$
so that, at a critical point $p=\bA(q)$, the Hessian is
$$T^2 H((U_{1},V_{1}),(U_{2},V_{2}))=2\langle V_{1}-T_{q}\bA(U_{1}),V_{2}-T_{q}\bA(U_{2})\rangle.$$
Let us notice that
$$T^2 H(\f_{4},\f_{5})=2\|\bB\|^{-1}\langle\bB\times \bc,\bB\times \bd\rangle=0,$$
$$T^2 H(\f_{4},\f_{6})=2\langle\bB\times \bc,\bb\rangle=0,$$
$$T^2 H(\f_{5},\f_{6})=2\langle\bB\times \bd,\bb\rangle=0.$$
The Hessian is diagonal in the basis $(\f_{4}, \f_{5}, \f_{6})$. Moreover we have
$$T^2 H(\f_{4},\f_{4})=\dx^2 H(\f_{5},\f_{5})=2\|\bB\|^{-1}\|\bB\times \bc\|^2=2\|\bB\|^{-1}\|\bB\times \bd\|^2=2\|\bB\|.$$
Finally we have:
$$T^2 H(\f_{6},\f_{6})=2.$$

Now we consider the local diffeomorphism:
$$(x,\xi)\mapsto \iota(x_{2},\xi_{2}, x_{3})+x_{1}\f_{4}(x_{2},\xi_{2}, x_{3})+\xi_{1}\f_{5}(x_{2},\xi_{2}, x_{3})+\xi_{3}\f_{6}(x_{2},\xi_{2}, x_{3}).$$
The Jacobian of this map is a symplectic matrix on $\Sigma$. We may apply the Moser-Weinstein argument (see \cite{weinstein-symplectic}) to make this map locally symplectic near $\Sigma$ modulo a change of variable which is tangent to the identity.

Near $\Sigma$, in these new coordinates, the Hamiltonian $H$ admits the expansion
\begin{equation}\label{hatH}
\hat H=H^0+\OO(|x_{1}|^3+|\xi_{1}|^3+|\xi_{3}|^3),
\end{equation}
where $\hat H$ denotes $H$ in the coordinates $(x_{1},x_{2}, x_{3},\xi_{1},\xi_{2},\xi_{3})$, 
and with
\begin{equation}\label{H0}
H^0=\xi_{3}^2+b(x_{2},\xi_{2},x_{3})(x_{1}^2+\xi_{1}^2), \quad b=\|\bB(x_{2},\xi_{2},x_{3})\|.
\end{equation}

\subsection{Semiclassical Birkhoff normal form}

\subsubsection{Birkhoff procedure in formal series}\label{subsec.Birkhoff}
Let us consider the space $\mathcal E$ of formal power series in
$(x_{1},\xi_{1},\xi_{3},\hbar)$ with coefficients smoothly
depending on $\tilde x=(x_{2}, \xi_{2}, x_{3})$: 
$$\mathcal
E=C^\infty_{x_{2}, \xi_{2},
  x_{3}}\formel{x_{1},\xi_{1},\xi_{3},\hbar}.$$ We endow $\mathcal E$
with the semiclassical Moyal product (with respect to all variables
$(x_{1},x_{2}, x_{3},\xi_{1},\xi_{2},\xi_{3})$) denoted by $\star$ and
the commutator of two series $\kappa_1$ and $\kappa_2$ is defined as
\[
[\kappa_1,\kappa_2]=\kappa_1\star\kappa_2-\kappa_2\star\kappa_1\,.
\]
The degree of
$x_{1}^{\alpha_1}\xi_{1}^{\alpha_2}\xi_{3}^\beta\hbar^\ell=z_{1}^{\alpha}\xi_{3}^\beta\hbar^\ell$
is $\alpha_1+\alpha_2+\beta+2\ell=|\alpha|+\beta+2\ell$. $\mathcal D_N$
denotes the space of monomials of degree $N$. $\mathcal O_N$ is the
space of formal series with valuation at least $N$. For any
$\tau,\gamma\in \mathcal E$, we denote
$\operatorname{ad}_\tau\gamma=[\tau,\gamma]$.

\begin{proposition}\label{formalBirkhoff}
Given $\gamma\in \mathcal O_3$, there exist formal power series
$\tau, \kappa\in \mathcal O_3$ such that
\[
e^{i\hbar^{-1}\operatorname{ad}_\tau}(H^0+\gamma)=H^0+\kappa\,,
\]
with $[\kappa, |z_1|^2]=0\,$.
\end{proposition}

\begin{proof}
Let $N\geq 1$. Assume that we have, for $\tau_N\in \mathcal O_3$,
\[
e^{i\hbar^{-1}\operatorname{ad}_{\tau_N}}(H^0+\gamma)=H^0+K_3+\cdots
+K_{N+1}+R_{N+2}+\mathcal O_{N+3}\,,
\]
with $K_i\in \mathcal D_i$, $[K_i, |z_1|^2]=0$ and $R_{N+2}\in
\mathcal D_{N+2}\,$.

Let $\tau^\prime\in \mathcal D_{N+2}$. Then we have
\[
e^{i\hbar^{-1}\operatorname{ad}_{\tau_N+\tau^\prime}}(H^0+\gamma)=H^0+K_3+\cdots
+K_{N+1}+K_{N+2}+\mathcal O_{N+3},
\]
with $K_{N+2}\in \mathcal D_{N+2}$ such that
\[
K_{N+2}=R_{N+2}+i\hbar^{-1}\operatorname{ad}_{\tau^\prime}H^0+\mathcal
O_{N+3}.
\]

\begin{lemma}
For $\tau^\prime \in \mathcal D_{N+2}$, we have
\[
i\hbar^{-1}\operatorname{ad}_{\tau^\prime}H^0=i\hbar^{-1}b\operatorname{ad}_{\tau^\prime}|z_1|^2+\mathcal
O_{N+3}.
\]
\end{lemma}

To prove this lemma, we observe that
\[
i\hbar^{-1}\operatorname{ad}_{\tau^\prime}H^0
=i\hbar^{-1}\operatorname{ad}_{\tau^\prime}\xi_{3}^2
+i\hbar^{-1}\operatorname{ad}_{\tau^\prime}(b(\tilde x)|z_1|^2).
\]
Let us write
\[
\tau^\prime = \sum_{|\alpha|+\beta+2l=N+2} a_{\alpha,\beta,l}(\tilde
x)z_{1}^{\alpha} \xi_{3}^\beta\hbar^l.
\]
Then, for the first term, we have
\[
\begin{aligned}
i\hbar^{-1}\operatorname{ad}_{\tau^\prime}\xi_{3}^2=&
\{\tau^\prime,\xi_{3}^2\}\\ =&
-2\xi_3\frac{\partial \tau^\prime}{\partial x_3}\\
=& -2 \sum_{|\alpha|+\beta+2\ell=N+2} \frac{\partial
a_{\alpha,\beta,\ell}}{\partial x_3}(\tilde x) z_{1}^{\alpha}
\xi_{3}^{\beta+1}\hbar^\ell\in \mathcal O_{N+3}\,.
\end{aligned}
\]
We also have
\[
\begin{aligned}
i\hbar^{-1}(\operatorname{ad}_{\tau^\prime}b(\tilde x))=&
\{\tau^\prime, b\}+\hbar^2\OO_{N}\\
=& \frac{\partial \tau^\prime}{\partial \xi_3}\frac{\partial
b}{\partial x_3}+\frac{\partial \tau^\prime}{\partial
\xi_2}\frac{\partial b}{\partial x_2}-\frac{\partial
\tau^\prime}{\partial x_2}\frac{\partial
b}{\partial \xi_2}+\OO_{N+1}\\
=& \sum_{|\alpha|+\beta+2\ell=N+2} \beta a(\tilde x)\frac{\partial b}{\partial x_3}
z_{1}^{\alpha}|z_1|^2 \xi_{3}^{\beta-1}\hbar^\ell+\mathcal O_{N+1} \in
\mathcal O_{N+1}\,.
\end{aligned}
\]
Therefore, for the second term, we get
\[
\begin{aligned}
i\hbar^{-1}\operatorname{ad}_{\tau^\prime}(b(\tilde x)|z_1|^2)=&
i\hbar^{-1}(\operatorname{ad}_{\tau^\prime}b(\tilde
x))|z_1|^2+i\hbar^{-1}b(\tilde
x)\operatorname{ad}_{\tau^\prime}|z_1|^2\\
=& i\hbar^{-1}b(\tilde
x)\operatorname{ad}_{\tau^\prime}|z_1|^2+\mathcal O_{N+3}\,,
\end{aligned}
\]
that completes the proof of the lemma.

By the lemma, we obtain that
\[
K_{N+2}=R_{N+2}+b\operatorname{ad}_{\tau^\prime}|z_1|^2,
\]
that we rewrite as
\[
R_{N+2}=K_{N+2}+i\hbar^{-1}b\operatorname{ad}_{|z_1|^2}\tau^\prime
=K_{N+2}+b\{|z_1|^2,\tau^\prime\}.
\]
Since $b(\tilde x)\neq 0$, we deduce the existence of $\tau^\prime$
and $K_{N+2}$ such that $K_{N+2}$ commutes with $|z_{1}|^2$.
\end{proof}
\subsubsection{Quantizing the formal procedure}
Let us now prove Theorem \ref{normal-form1}. Using \eqref{hatH} and
applying the Egorov theorem (see \cite{Ro87, Z13} or
Theorem~\ref{theo:egorov2}), we can find a unitary Fourier Integral
Operator $U_\hbar$, and such that
\[
U^*_\hbar \mathcal{L}_{\hbar,\bA} U_\hbar= C_0\hbar +
\Op_{\hbar}^w (H^0)+\Op^w_\hbar(r_\hbar),
\]
where the Taylor series (with respect to $x_{1}$, $\xi_{1}$, $\xi_3$,
$\hbar$) of $r_{\hbar}$ satisfies $r_\hbar^T=\gamma \in \mathcal O_3$
and $C_0$ is the value at the origin of the sub-principal symbol of
$U^*_\hbar \mathcal{L}_{\hbar,\bA} U_\hbar$. One can choose $U_\hbar$
such that the subprincipal symbol is preserved by
conjugation\footnote{This is sometimes called the \emph{Improved
    Egorov Theorem}. It was first discovered by Weinstein
  in~\cite{weinstein-maslov}, in the homogeneous setting. For the
  semiclassical case, see for instance \cite[Appendix~A]{HelSj89}.},
which implies $C_0=0$. Applying Proposition~\ref{formalBirkhoff}, we
obtain $\tau$ and $\kappa$ in $\mathcal O_3$ such that
\[
e^{i\hbar^{-1}\operatorname{ad}_\tau}(H^0+\gamma)=H^0+\kappa,
\]
with $[\kappa, |z_1|^2]=0$.

We can introduce a smooth symbol $a_{\h}$ with compact support such
that we have $a_{\h}^T=\tau$ in a neighborhood of the origin. By
Proposition~\ref{formalBirkhoff} and Theorem~\ref{theo:egorov4}, we
obtain that the operator 
\[
e^{i\hbar^{-1}\Op^w_\hbar(a_{\h})} (\Op_{\hbar}^w
(H^0)+\Op^w_\hbar(r_\hbar))e^{-i\hbar^{-1}\Op^w_\hbar(a_{\h})}
\] 
is a pseudodifferential operator such that the formal Taylor series of
its symbol is $H^0+\kappa$. In this application of
Theorem~\ref{theo:egorov4}, we have used the filtration $\OO_{j}$
defined in Section~\ref{subsec.Birkhoff}.  Since $\kappa$ commutes
with $|z_1|^2$, we can write it as a formal series in $|z_1|^2$:
\[
\kappa=\sum_{k\geq 3}\sum_{2\ell+2m+\beta=k}\hbar^\ell
c_{\ell,m}(x_2,\xi_2,x_3)|z_1|^{2m}\xi_3^\beta.
\]
This formal series can be reordered by using monomials
$(|z_1|^{2})^{\star m}$:
\[
\kappa=\sum_{k\geq 3}\sum_{2\ell+2m+\beta=k}\hbar^\ell
c^\star_{\ell,m}(x_2,\xi_2,x_3)(|z_1|^{2})^{\star m}\xi_3^\beta.
\]
Thanks to the Borel lemma, we may find a smooth function, with a compact support as small as we want with respect to $\hbar$, $I$ and $\xi_{3}$, $f^\star(\hbar,I, x_{2}, \xi_{2}, x_{3}, \xi_{3})$ such that its Taylor series with respect to $\hbar, I, \xi_{3}$ is
\[
\sum_{k\geq 3}\sum_{2\ell+2m+\beta=k}\hbar^\ell
c^\star_{\ell,m}(x_2,\xi_2,x_3) I^{m}\xi_3^\beta.
\]
This achieves the proof of Theorem \ref{normal-form1}.
\subsection{Spectral reduction to the first normal form}
This section is devoted to the proof of Corollary \ref{normal-form1-d}.

\subsubsection{Numbers of eigenvalues}\label{subsec.numbers}
\begin{lemma}\label{lem.spess}
Under Assumption \ref{hyp3-4}, there exists $h_{0}>0$ and $\eps_{0}>0$ such that for all $\hbar\in(0,h_{0})$,
$$\inf\mathfrak{s}_{\mathsf{ess}}(\mathcal{N}_{\hbar}^\sharp)\geq (\beta_{0}+\eps_{0})\hbar.$$
\end{lemma}
\begin{proof}
By using the assumption we may consider a smooth function $\chi$ with compact support and $\eps_{0}>0$ such that 
$$\xi_{3}^2+b(x_{2}, \xi_{2},x_{3})+\chi(x_{2}, x_{3}, \xi_{2}, \xi_{3})\geq \beta_{0}+2\eps_{0}.$$
Then, given $\eta\in(0,1)$ and estimating the second term in \eqref{Nheta} by using that the support of $f^\star$ is chosen small enough and the semiclassical Calderon-Vaillancourt theorem, we notice that, for $\hbar$ small enough,
\begin{equation}\label{lb-Nh}
\mathcal{N}_{\hbar}^\sharp\geq (1-\eta)\Op_{\hbar}^w\left(\xi_{3}^2+|z_{1}|^2b(x_{2}, \xi_{2},x_{3})\right).
\end{equation}
Since the essential spectrum is invariant by (relatively) compact perturbations, we have
$$\mathfrak{s}_{\mathsf{ess}}\left(\mathcal{N}_{\hbar}^\sharp+ (1-\eta)\hbar\Op_{\hbar}^w\chi(x_{2}, x_{3}, \xi_{2}, \xi_{3})\right)=\mathfrak{s}_{\mathsf{ess}}\left(\mathcal{N}_{\hbar}^\sharp\right).$$
Hence 
$$\inf\mathfrak{s}_{\mathsf{ess}}\left(\mathcal{N}_{\hbar}^\sharp\right)\geq \inf \mathfrak{s}\left(\mathcal{N}_{\hbar}^\sharp+ (1-\eta)\hbar\Op_{\hbar}^w\chi(x_{2}, x_{3}, \xi_{2}, \xi_{3})\right).$$
In order to bound the r.h.s. from below, we write
\begin{align*}
&\mathcal{N}_{\hbar}^\sharp+ (1-\eta)\hbar\Op_{\hbar}^w\chi(x_{2}, x_{3}, \xi_{2}, \xi_{3})\\
&\geq (1-\eta)\Op_{\hbar}^w\left(\xi_{3}^2+|z_{1}|^2b(x_{2}, \xi_{2},x_{3})\right)+ (1-\eta)\hbar\Op_{\hbar}^w\chi(x_{2}, x_{3}, \xi_{2}, \xi_{3})\\
&\geq  \hbar(1-\eta)\Op_{\hbar}^w\left(\xi_{3}^2+b(x_{2}, \xi_{2},x_{3})+\chi(x_{2}, x_{3}, \xi_{2}, \xi_{3})\right)\\
&\geq  \hbar(1-\eta)(\beta_{0}+2\eps_{0}-C\hbar),
\end{align*}
where we have used the semiclassical G\aa rding inequality. Taking $\eta$ and then $\hbar$ small enough, this concludes the proof.
\end{proof}
By using the Hilbertian decomposition given by the Hermite functions $(e_{k,\hbar})_{k\geq 1}$ associated with $\mathcal{I}_{\hbar}$, we notice that
$$\mathcal{N}_{\hbar}^\sharp=\bigoplus_{k\geq 1} \mathcal{N}_{\hbar}^{[k],\sharp},$$
where
\begin{equation}\label{Nhk}
\mathcal{N}_{\hbar}^{[k],\sharp}=\hbar^2D_{x_{3}}^2+(2k-1)\hbar\Op_{\hbar}^w b+\Op_{\hbar}^w f^{\star,\sharp}(\hbar,(2k-1)\hbar, x_{2}, \xi_{2}, x_{3}, \xi_{3}),
\end{equation}
acting on $L^2(\R^2)$.
\begin{lemma}
For all $\eta\in(0,1)$, there exist $C>0$ and $h_{0}>0$ such that for all $k\geq 1$ and $\hbar\in(0,h_{0})$, we have $\mathfrak{s}_{1}(\mathcal{N}_{\hbar}^{[k],\sharp})\geq (1-2\eta)b_{0}(2k-1)\hbar$.
\end{lemma}
\begin{proof}
Applying \eqref{lb-Nh}
to $\psi(x_{1}, x_{2}, x_{3})=\varphi(x_{2}, x_{3}) e_{k,\hbar}(x_{1})$, we infer that
$$\langle\mathcal{N}^{[k],\sharp}_{\hbar}\varphi,\varphi\rangle\geq (2k-1)\hbar (1-\eta) \langle\Op_{\hbar}^w (b) \varphi,\varphi\rangle.$$
With the G\aa rding inequality, we get
$$\langle\Op_{\hbar}^w (b) \varphi,\varphi\rangle\geq (b_{0}-C\hbar)\|\varphi\|^2\,,$$
and the conclusion follows by the min-max principle.
\end{proof}
We immediately deduce the following proposition.
\begin{proposition}\label{finite-k}
We have the following descriptions of the low lying spectrum of $\mathcal{N}_{\hbar}^\sharp$.
\begin{enumerate}[(a)]
\item There exist $\hbar_{0}>0$ and $K\in\N$ such that, for $\hbar\in(0,\hbar_{0})$, the spectrum of $\mathcal{N}_{\hbar}^\sharp$ lying below $\beta_{0}\hbar$ is contained in the union $\bigcup_{k=1}^K\spe\left(\mathcal{N}_{\hbar}^{[k],\sharp}\right)$.
\item If $c\in(0,\min(3b_{0},\beta_{0}))$, then there exists $\hbar_{0}>0$ such that for all $\hbar\in(0,\hbar_{0})$ the eigenvalues of $\mathcal{N}_{\hbar}^\sharp$ lying below $c\hbar$ coincide with the eigenvalues of $\mathcal{N}_{\hbar}^{[1],\sharp}$ below $c\hbar$.
\end{enumerate}
\end{proposition}
We deduce the following proposition.
\begin{corollary}\label{numbers-L-N}
Under Assumption \ref{hyp3}, we have
$$\mathsf{N}\left(\mathcal{L}_{\hbar,\bA},\beta_{0}\hbar\right)=\OO(\hbar^{-3/2}),\quad \mathsf{N}\left(\mathcal{N}_{\hbar}^\sharp,\beta_{0}\hbar\right)=\OO(\hbar^{-2}).$$
\end{corollary}
\begin{proof}
To get the first estimate, we use the Lieb-Thirring inequalities (which provide an upper bound of the number of eigenvalues in dimension three) and the diamagnetic inequality (see \cite{RVN13} and Proposition \ref{lbB}). To get the second estimate, we use the first point in Proposition \ref{finite-k}.
Moreover, given $\eta\in(0,1)$, by using $\hbar\in(0,1)$ we infer
$$\langle\mathcal{N}_{\hbar}^{[k], \sharp}\psi,\psi\rangle\geq (1-\eta)\hbar \langle\Op^w_{\hbar} \left(\xi_{3}^2+b(x_{2},\xi_{2}, x_{3})\right)\psi,\psi\rangle.$$
Note that the last inequality is very rough. By the min-max principle, we deduce that 
$$\mathsf{N}\left(\mathcal{N}_{\hbar}^{[k],\sharp},\beta_0 \hbar\right)\leq  \mathsf{N}\left(\Op^w_{\hbar} \left(\xi_{3}^2+b(x_{2},\xi_{2}, x_{3})\right),(1-\eta)^{-1}\beta_0 \right).$$
Then, we conclude by using the Weyl asymptotics and our confinement assumption:
$$\mathsf{N}\left(\Op^w_{\hbar} \left(\xi_{3}^2+b(x_{2},\xi_{2}, x_{3})\right),(1-\eta)^{-1}\beta_0 \right)= \OO(\hbar^{-2}).$$
\end{proof}
Since $\mathcal{N}_{\hbar}^\sharp$ commutes with $\mathcal{I}_{\hbar}$, we also deduce the following corollary.
\begin{corollary}\label{tensorx1}
For any eigenvalue $\lambda$ of $\mathcal{N}_{\hbar}^\sharp$ such that $\lambda\leq \beta_{0}\hbar$ we may consider an orthonormal eigenbasis of the space $\ker\left(\mathcal{N}^\sharp_{\hbar}-\lambda\right)$ formed with functions in the form $e_{k,\hbar}(x_{1})\varphi_{\hbar}(x_{2},x_{3})$ with $k\in\{1,\ldots K\}$. Moreover we have $\mathds{1}_{(-\infty,\beta_{0}\hbar)}(\mathcal{N}_{\hbar}^\sharp)=\OO(\hbar^{-2})$ and each eigenfunction associated with $\lambda\leq \beta_{0}\hbar$ is a linear combination of  at most $\OO(\hbar^{-2})$ such tensor products.
\end{corollary}

\subsubsection{Microlocalization estimates}
The following proposition follows from the same lines as in dimension two (see \cite[Theorem 2.1]{HelMo96}).
\begin{proposition}
Under Assumptions \ref{hyp1-2} and \ref{hyp3-4}, for any $\epsilon >0$, there exist $C(\epsilon) >0$ and $h_{0}(\epsilon) >0$ such that for any eigenpair $(\lambda,\psi)$ of $\mathcal{L}_{\hbar,\bA}$ with $\lambda\leq \beta_0\, \hbar$ we have for $\hbar\in(0,h_{0}(\epsilon))$:
$$\int_{\R^3} e^{ 2 (1-\epsilon) \phi(q)/\h^{\frac 12}}|\psi|^2\, \dx q  \leq C(\epsilon) \exp( \epsilon \h^{-\frac 12}) \|\psi\|^2,$$ 
$$\mathcal{Q}_{\hbar,\bA}(e^{(1-\epsilon) \phi(q)/\h^{\frac 12}}\psi)\leq C(\epsilon) \exp( \epsilon \h^{-\frac 12}) \|\psi\|^2,$$
where $\phi$ is the distance to the bounded set $\{ \|\bB(q)\| \leq \beta_0\} $ for the Agmon metric
 $(\|(\bB(q)\| -\beta_0)_+ g$, with $g$ the standard metric.
\end{proposition}

\begin{proposition}
Under Assumptions \ref{hyp1-2} and \ref{hyp3-4}, we consider $0<b_{0}<\beta_0 < b_1$ and there exist $C>0$ and $\h_{0}>0$ such that for any eigenpair $(\lambda,\psi)$ of $\mathcal{L}_{\hbar,\bA}$ with $\lambda\leq \beta_{0}\h$ we have for $\h\in(0,\h_{0})$ and $\delta \in (0,\frac 12)$:
$$\psi=\chi_{0}\left(\hbar^{-2\delta}\mathcal{L}_{\hbar,\bA}\right)\chi_{1}(q)\psi+\OO(\hbar^{\infty})\|\psi\|,$$
where $\chi_{0}$ is a cutoff function compactly supported in the ball of center $0$ and radius $1$ and where $\chi_{1}$ is  a compactly supported smooth cutoff function being $1$ in an open  neighborhood of $\{ \|\bB(q)\| \leq \beta_0\} $.
\end{proposition}
Let us now investigate the microlocalization of the eigenfunctions of $\mathcal{N}_{\hbar}^\sharp$.
\begin{proposition}\label{space-loc-N}
Let $\chi$ be a smooth cutoff function being $0$ on $\{b\leq\beta_{0}\}$ and $1$ on the set $\{b\geq\beta_{0}+\eps\}$. If $\lambda$ is an eigenvalue of $\mathcal{N}_{\hbar}^\sharp$ such that $\lambda\leq \beta_{0}\hbar$ and if $\psi$ is an associated eigenfunction, then we have
$$\Op_{\hbar}^w\left(\chi(x_{2},\xi_{2},x_{3})\right)\psi=\OO(\hbar^{\infty})\|\psi\|.$$
\end{proposition}
\begin{proof}
Due to Corollary \ref{tensorx1}, it is sufficient to prove the estimate for a function in the form $\psi(x_{1},x_{2},x_{3})=e_{k,\hbar}(x_{1})\varphi(x_{2},x_{3})$ where $k$ lies in $\{1,\ldots,K\}$ and we have
$$\mathcal{N}^{\sharp}_{\hbar}\psi=\lambda\psi,\qquad \mbox{ or equivalently }\qquad\mathcal{N}^{[k], \sharp}_{\hbar}\varphi=\lambda\varphi,$$
where we recall \eqref{Nhk}. Then, we write
$$\mathcal{N}^{[k], \sharp}_{\hbar}\Op_{\hbar}^w\left(\chi\right)\varphi=\lambda\Op_{\hbar}^w\left(\chi\right)\varphi+\left[\mathcal{N}^{[k], \sharp}_{\hbar},\Op_{\hbar}^w\left(\chi\right)\right]\varphi$$
and it follows that
\begin{multline}\label{variational-eve}
\left\langle\mathcal{N}^{[k], \sharp}_{\hbar}\Op_{\hbar}^w\left(\chi\right)\varphi,\Op_{\hbar}^w\left(\chi\right)\varphi \right\rangle=\lambda\|\Op_{\hbar}^w\left(\chi\right)\varphi\|^2\\
+\left\langle \left[\mathcal{N}^{[k], \sharp}_{\hbar},\Op_{\hbar}^w\left(\chi\right)\right]\varphi, \Op_{\hbar}^w\left(\chi\right)\varphi\right\rangle.
\end{multline}
Rough pseudo-differential estimates imply that there exist $C>0$, $\hbar_{0}>0$ such that for all $\hbar\in(0,\hbar_{0})$,
\begin{multline}\label{commutatorNhk}
\left|\left\langle \left[\mathcal{N}^{[k], \sharp}_{\hbar},\Op_{\hbar}^w\left(\chi\right)\right]\varphi, \Op_{\hbar}^w\left(\chi\right)\varphi\right\rangle\right|\leq C\hbar^2\left\| \Op_{\hbar}^w\left(\underline{\chi}\right)\varphi\right\|^2+C\hbar\left\| \Op_{\hbar}^w\left(\underline{\chi}\right)\varphi\right\|^2\\
+C\hbar\langle \Op_{\hbar}^w\left(\partial_{3}\chi\right)\varphi,\Op_{\hbar}^w\left(\xi_{3}\right) \Op_{\hbar}^w\left(\chi\right)\varphi\rangle.
\end{multline}
Combining \eqref{commutatorNhk} and \eqref{variational-eve}, we get
\begin{equation}\label{xi3chi}
\left\|\Op_{\hbar}^w\left(\xi_{3}\right) \Op_{\hbar}^w\left(\chi\right)\varphi\right\|\leq C\hbar^{\frac{1}{2}}\| \Op_{\hbar}^w\left(\underline{\chi}\right)\varphi\|\,,
\end{equation}
where $\underline{\chi}$ is a smooth cutoff function living on a slightly larger support than $\chi$. By using \eqref{xi3chi}, we can improve the commutator estimate
$$\left|\left\langle \left[\mathcal{N}^{[k], \sharp}_{\hbar},\Op_{\hbar}^w\left(\chi\right)\right]\varphi, \Op_{\hbar}^w\left(\chi\right)\varphi\right\rangle\right|\leq C\hbar^{\frac{3}{2}}\left\| \Op_{\hbar}^w\left(\underline{\chi}\right)\varphi\right\|^2.$$
We infer  that, there exist $C>0$, $\hbar_{0}>0$ such that for $\hbar\in(0,\hbar_{0})$,
$$\left\langle\mathcal{N}^{[k], \sharp}_{\hbar}\Op_{\hbar}^w\left(\chi\right)\varphi,\Op_{\hbar}^w\left(\chi\right)\varphi \right\rangle\leq \beta_{0}\hbar\| \Op_{\hbar}^w\left(\chi\right)\varphi\|^2+C\hbar^{\frac{3}{2}}\left\| \Op_{\hbar}^w\left(\underline{\chi}\right)\varphi\right\|^2.$$
By using the semiclassical G\aa rding inequality and the support of $\chi$, we get
$$\left\langle\mathcal{N}^{[k], \sharp}_{\hbar}\Op_{\hbar}^w\left(\chi\right)\varphi,\Op_{\hbar}^w\left(\chi\right)\varphi \right\rangle\geq (\beta_{0}+\eps_{0})\hbar\left\|\Op_{\hbar}^w\left(\chi\right)\varphi\right\|^2$$
and we deduce
$$\left\|\Op_{\hbar}^w\left(\chi\right)\varphi\right\|^2\leq C\hbar^{\frac{1}{2}}\left\| \Op_{\hbar}^w\left(\underline{\chi}\right)\varphi\right\|^2.$$
The conclusion follows by a standard iteration argument.
\end{proof}
The following proposition is concerned by the microlocalization with respect to $\xi_{3}$.
\begin{proposition}\label{xi3-loc-N}
Let $\chi_{0}$ be a smooth cutoff function being $0$ in a neighborhood of $0$ and let $\delta\in\left(0,\frac{1}{2}\right)$.  If $\lambda$ is an eigenvalue of $\mathcal{N}_{\hbar}^\sharp$ such that $\lambda\leq \beta_{0}\hbar$ and if $\psi$ is an associated eigenfunction, then we have
$$\Op_{\hbar}^w\left(\chi_{0}\left(\hbar^{-\delta}\xi_{3}\right)\right)\psi=\OO(\hbar^{\infty})\|\psi\|.$$
\end{proposition}
\begin{proof}
We write again $\psi(x_{1},x_{2},x_{3})=e_{k,\hbar}(x_{1})\varphi(x_{2},x_{3})$ with $k\in\{1,\ldots,K\}$ and we have $\mathcal{N}^{[k], \sharp}_{\hbar}\varphi=\lambda\varphi$. We use again the formula \eqref{variational-eve} with $\chi_{0}\left(\hbar^{-\delta}\xi_{3}\right)$. We get the commutator estimate
\begin{multline*}
\left|\left\langle \left[\mathcal{N}^{[k], \sharp}_{\hbar},\Op_{\hbar}^w\left(\chi_{0}\left(\hbar^{-\delta}\xi_{3}\right)\right)\right]\varphi, \Op_{\hbar}^w\left(\chi_{0}\left(\hbar^{-\delta}\xi_{3}\right)\right)\varphi\right\rangle\right|\\
\leq C\hbar^{\frac{3}{2}-\delta}\left\| \Op_{\hbar}^w\left(\underline{\chi}_{0}\left(\hbar^{-\delta}\xi_{3}\right)\right)\varphi\right\|^2.
\end{multline*}
We have
$$\Op_{\hbar}^w\left( (\hbar^{-\delta}\xi_{3})^2\chi^2_{0}\left(\hbar^{-\delta}\xi_{3}\right)\right)=\Op_{\hbar^{1-\delta}}^w\left( \xi_{3}^2\chi^2_{0}\left(\xi_{3}\right)\right),$$
so that, with the G\aa rding inequality,
$$\left\langle\Op_{\hbar}^w\left( (\hbar^{-\delta}\xi_{3})^2\chi^2_{0}\left(\hbar^{-\delta}\xi_{3}\right)\right)\varphi, \varphi\right\rangle\geq (1-C\hbar^{1-\delta})\|\varphi\|^2.$$
We infer
\begin{multline*}
(\hbar^{2\delta}(1-Ch^{1-\delta})-\beta_{0}\hbar)\left\|\Op_{\hbar}^w\left(\chi_{0}\left(\hbar^{-\delta}\xi_{3}\right)\right)\varphi\right\|^2\\
\leq C\hbar^{\frac{3}{2}-\delta}\left\| \Op_{\hbar}^w\left(\underline{\chi}_{0}\left(\hbar^{-\delta}\xi_{3}\right)\right)\varphi\right\|^2.
\end{multline*}
\end{proof}
Using $\Op_{\hbar}^w f^\star(\hbar,\mathcal{I}_{\hbar}, x_{2}, \xi_{2}, x_{3}, \xi_{3})=\Op_{\hbar}^w f(\hbar,|z_{1}|^2, x_{2}, \xi_{2}, x_{3}, \xi_{3})$, we deduce the following in the same way.
\begin{proposition}\label{z1-loc-N}
Let $\chi_{1}$ be a smooth cutoff function being $0$ in a neighborhood of $0$ and let $\delta\in\left(0,\frac{1}{2}\right)$.  If $\lambda$ is an eigenvalue of $\mathcal{N}_{\hbar}^\sharp$ such that $\lambda\leq \beta_{0}\hbar$ and if $\psi$ is an associated eigenfunction, then we have
$$\Op_{\hbar}^w\left(\chi_{1}\left(\hbar^{-\delta}(x_{1},\xi_{1})\right)\right)\psi=\OO(\hbar^{\infty})\|\psi\|.$$
\end{proposition}

\begin{proposition}
The spectra of $\mathcal{L}_{\hbar,\bA}$ and  $\mathcal{N}_{\hbar}^\sharp$ below $\beta_0 \hbar$ coincide modulo $\OO(\hbar^\infty)$.
\end{proposition}
\begin{proof}
We refer to \cite[Section 4.3]{RVN13} which contains similar arguments.
\end{proof}
This proposition provides the point \eqref{normal-form1-d-a} in Corollary \ref{normal-form1-d}. With Proposition~\ref{finite-k}, we deduce the point \eqref{normal-form1-d-b}.

\section{Second Birkhoff normal form}\label{BNF2}

\subsection{Birkhoff analysis of the first level}
This section is devoted to the proofs of Theorems \ref{pre-normal-form2} and \ref{normal-form2}.

The goal now is to normalize a $\hbar$-pseudo-differential
operator $\mathcal{N}_{\hbar}^{[1]}$ on $\R^2$ whose Weyl symbol has the form
\[
N^{[1]}_{\hbar} = \xi_3^2 + \h b(x_2,\xi_2,x_3) + r_\h(x_2,\xi_2,x_3,\xi_3),
\]
where $r_\h$ is a classical symbol with the following asymptotic expansion:
\[
r_\h = r_0 + \h r_1 + \h^2 r_2 + \cdots
\]
(in the symbol class topology), where each $r_\ell$ has a formal expansion in $\xi_3$ of the form
\begin{equation}
r_\ell(x_2,\xi_2,x_3,\xi_3) \sim \sum_{2\ell+\beta \geq 3}
c_{\ell,\beta}(x_2,\xi_2,x_3) \xi_3^\beta.\label{eq:series}
\end{equation}
The leading terms of $N^{[1]}_{\hbar}$ are:
\begin{equation}
  N^{[1]}_{\hbar} = \xi_3^2 + \h b(x_2,\xi_2,x_3) + c_{1,1}(x_2,\xi_2,x_3)\h \xi_3 +
  \OO(\h\xi_3^2) + \OO(\xi_3^3) + \OO(\h^2).
\label{eq:N}
\end{equation}

\subsubsection{First normalization of the symbol}
We consider the following local change of variables $\hat{\varphi}(x_2,\xi_2,x_3,\xi_3) = (\hat x_2, \hat \xi_2, \hat x_3, \hat\xi_3)$:
\begin{equation}
  \left\{\begin{aligned}
    \label{eq:non-symplectic}
    \hat x_2 & := x_2 + \xi_3 \partial_2 s(x_2,\xi_2)\,,\\
    \hat \xi_2 & := \xi_2 +\xi_3 \partial_1 s(x_2,\xi_2)\,,\\
    \hat x_3 & := x_3 - s(x_2,\xi_2)\,, \\
    \hat \xi_3 & := \xi_3\,.
  \end{aligned}
\right.
\end{equation}
It is easy to check that the differential of $\hat{\varphi}$ is
invertible as soon as $\xi_3$ is small enough. Moreover, we have
\[
\hat{\varphi}^* \omega_0 - \omega_0 = \OO(|\xi_3|).
\]
By the Darboux-Weinstein theorem (see for instance~\cite[Lemma
2.4]{RVN13}), there exists a local diffeomorphism $\psi$ such that
\begin{equation}
  \psi = \textup{Id} + \OO(\xi_3^2) \quad \text{ and } \quad \psi^* \hat{\varphi}^* \omega_0 = \omega_0.
  \label{eq:darboux}
\end{equation}

Using the improved Egorov theorem, one can find a unitary Fourier Integral Operator $V_\h$
such that the Weyl symbol of $V_\h^* \mathcal{N}_{\hbar}^{[1]} V_\h$
is $\hat N_{\hbar} := N^{[1]}_{\hbar}\circ\hat{\varphi}\circ\psi +
\OO(\h^2)$. From \eqref{eq:darboux}, and~\eqref{eq:non-symplectic}, we
see that $\hat{r}_\h := r_\h\circ \hat{\varphi}\circ\psi$ is still of
the form~\eqref{eq:series}, with modified coefficients
$c_{\ell,\beta}$. Thus, using the new variables and a Taylor expansion
in $\xi_3$, we get
\begin{multline*}
  \hat N_{\hbar}  = \hat\xi_3^2 + \h b(\hat x_2 + \OO(\hat\xi_3),\hat\xi_2 + \OO(\xi_3),\hat x_3 + s(\hat x_2 + \OO(\hat\xi_3),\hat\xi_2 + \OO(\hat\xi_3))  + \OO(\hat\xi_3^2))\\ 
 + \OO(\hat\xi_3^3) + \hat r_{\hbar}+\OO(\h^2)
\end{multline*}
and thus
\begin{multline} \label{eq:hatN}
 \hat N_{\hbar}    = \hat\xi_3^2 + \h b(\hat x_2, \hat \xi_2, \hat x_3 + s(\hat x_2,
  \hat \xi_2)) + \h\hat\xi_3g(\hat x_2, \hat \xi_2, \hat x_3) \\
  +\OO(\h\hat\xi_3^2)  + \hat r_\h + \OO(\hat\xi_3^3) + \OO(\h^2),
\end{multline}
for some smooth function $g(\hat x_2, \hat \xi_2, \hat x_3)$.\\
 Therefore $\hat N_{\hbar}$ has the following form:
\[ 
\hat N_{\hbar}  = \hat\xi_3^2 + \h b(\hat x_2, \hat \xi_2, \hat x_3 + s(\hat x_2,
  \hat \xi_2)) +  \hat c_{1,1}(x_2,\hat\xi_2,\hat x_3)\h \hat\xi_3 +
  \OO(\h\hat\xi_3^2) + \OO(\hat\xi_3^3) + \OO(\h^2).
\]

\subsubsection{Where the second harmonic oscillator appears} We now drop all the hats off the variables. We use a Taylor expansion with respect to $x_3$,
which, in view of~\eqref{eq:critical}, yields:
\[
b(x_2, \xi_2, x_3 + s(x_2, \xi_2)) =
b(x_2,\xi_2,s(x_2,\xi_2)) + \frac{x_3^2}2\partial_3^2
b(x_2,\xi_2,s(x_2,\xi_2)) + \OO(x_3^3).
\]
We let: 
\begin{equation}\label{def:beta}
\nu = (\tfrac{1}{2}\partial_3^2
b(x_2,\xi_2,s(x_2,\xi_2)))^{1/4}\mbox{ and } \gamma=\ln \nu.
\end{equation}
We introduce the change of coordinates $(\check x_{2}, \check x_{3},
\check\xi_{2}, \check\xi_{3})=C(x_{2},x_{3},\xi_{2},\xi_{3})$ defined
by:
\begin{equation}
\left\{\begin{array}{ll}
\check  x_{3}  &= \nu x_{3}\,,\\
\check  \xi_{3}& =\nu^{-1} \xi_{3}\,,\\
\check x_{2}   &=x_{2}+\frac{\dr \gamma}{\dr\xi_{2}}x_{3}\xi_{3}\,,\\
\check \xi_{2}   &=\xi_{2}-\frac{\dr \gamma}{\dr x_{2}}x_{3}\xi_{3}\,,\\
\end{array}\right.\label{eq:check}
\end{equation}
for which one can check that $C^*\omega_0 - \omega_0 = \OO(x_3\xi_3) =
\OO(\xi_3)$. As before, we can make this local diffeomorphism
symplectic by the Darboux-Weinstein theorem, which
modifies~\eqref{eq:check} by $\OO(\xi_3^2)$. In the new variables
(which we call $(x_{2},x_{3},\xi_{2},\xi_{3})$ again), the symbol $\check N_{\hbar}$
has the form:
\begin{align*}
  \check N_{\hbar} & = \nu^2(x_2,\xi_2)\left(\xi_3^2 + \h x_3^2\right) + \h
  b(x_2,\xi_2,s(x_2,\xi_2)) + \check c_{1,1}(x_2,\xi_2,x_3)\h \xi_3 \nonumber \\
  &  + \OO(\h x_3^3) +\OO(\h\xi_3^2) + \OO(\xi_3^3)
  + \OO(\h^2),\label{eq:N-final}
\end{align*}
for some smooth function $\check c_{1,1}(x_2,\xi_2,x_3)$.

\subsubsection{Normalizing the remainder}
 The next step is to get rid of the term
$\check c_{1,1}(x_2,\xi_2,x_3)\h \xi_3\,$.  Let
\[
a(x_2,\xi_2,x_3) := -\frac{1}{2}\int_0^{x_3} \check c_{1,1}(x_2,\xi_2,t)dt\,.
\]
Since $\check c_{1,1}$ is compactly supported, $a$ is bounded, and one can
form the unitary pseudo-differential operator $\exp(iA)$, $A=\opweyl(a)$. We have
\[
\exp (-iA)  \Op_{\hbar}^w\left(\check N_{\hbar}\right) \exp (iA) = \Op_{\hbar}^w\left(\check N_{\hbar}\right) + \exp(-iA)[ \Op_{\hbar}^w\left(\check N_{\hbar}\right), \exp(iA)].
\]
The symbol of $[\exp(-iA) \Op_{\hbar}^w\left(\check N_{\hbar}\right), \exp(iA)]$ is
$$\frac{\h}{i}e^{-ia}\{N, e^{ia}\} + \OO(\h^2) = \h
\{\check N_{\hbar},a\} + \OO(\h^2) = \h\{\check N_0,a\} + \OO(\h^2),$$ where $\check N_0$
is the principal symbol of $\check N_{\hbar}$, which satisfies:
\[
\check N_0 = \xi_3^2 + \OO(\xi_3^3).
\]
Therefore $\{\check N_{\hbar},a\} = \{\xi_3^2, a\} + \OO(\xi_3^2)$. Since 
\[
\{\xi_3^2, a\} = 2 \xi_3 \deriv{a}{x_3} =  - \xi_3 \check c_{1,1},
\]
we get
\[
\exp (-iA) \Op_{\hbar}^w\left(\check N_{\hbar}\right) \exp (iA) =  \opweyl(\check N_{\hbar}  - \h\xi_3 \check c_{1,1} + \OO(\h\xi_3^2) + \OO(\h^2)),
\]
which shows that we can remove the coefficient of $\h \xi_3$. The new operator given by the conjugation formula $\underline{\mathcal{N}}^{[1]}_{\hbar}= \exp (-iA)  \Op_{\hbar}^w\left(\check N_{\hbar}\right) \exp (iA) $ has a symbol of the form
\begin{equation}
  \underline{N}^{[1]}_{\hbar}  = \nu^2(x_2,\xi_2)\left(\xi_3^2 + \h x_3^2\right) + \h
  b(x_2,\xi_2,s(x_2,\xi_2)) 
   + \underline{r}_\h,\label{eq:N-final}
\end{equation}
where $\underline{r}_\h=\OO(\h x_3^3) +\OO(\h\xi_3^2) + \OO(\xi_3^3) +
\OO(\h^2)$.

This proves Theorem \ref{pre-normal-form2}.

\subsubsection{The second Birkhoff normal form}\label{BNF2b}
We now want to perform a Birkhoff normal form for $\underline{\mathcal
  N}^{[1],\sharp}_{\hbar}$ relative to the ``second harmonic oscillator''
$$\underline{\nu}^2(x_2,\xi_2)\left(\xi_3^2 + \h
  x_3^2\right).$$ 
Using Notation \ref{notation}, we introduce the new
semiclassical parameter $h=\h^{\frac{1}{2}}$, and use the relation
\[
\opweyl(\underline{N}^{[1],\sharp}_{\hbar}) = {\Op_{h}^w}(\underline{\mathsf N}_{h}^{[1],\sharp}).
\]
Thus, let $\tilde\xi_j := \h^{-1/2}\xi_j$. 
The new symbol $\underline{\mathsf N}_{h}^{[1],\sharp}$ has the form:
\begin{multline*}
 \underline{\mathsf N}_{h}^{[1],\sharp}(x_2,\tilde \xi_2,x_3,\tilde \xi_3) = 
  h^2\left(\underline{\nu}^2(x_2,h \tilde \xi_2)(\tilde\xi_3^2 + x_3^2) +
    \underline{b}(x_2, h\tilde \xi_2,s(x_2, h\tilde\xi_2))\right.\\
+\left. h^{-2} \underline{r}_{ h^2}^{\sharp}(x_2, h \tilde
    \xi_2,x_3, h\tilde \xi_3)\right).
\end{multline*}
We introduce momentarily a new parameter $\mu$ and define
\begin{multline*}
  \underline{\mathsf N}_{h}^{[1],\sharp} (x_2,\tilde \xi_2,x_3,\tilde
  \xi_3;\mu) := \underline{\nu}^2(x_2,\mu \tilde \xi_2)(\tilde\xi_3^2
  + x_3^2) +
  \underline{b}(x_2,\mu\tilde \xi_2,s(x_2,\mu\tilde\xi_2))\\
  + h^{-2}\underline{r}_{h^2}^{\sharp}(x_2,\mu \tilde \xi_2,x_3,
  h\tilde \xi_3).
\end{multline*}
Notice that $\underline{\mathsf N}_{h}^{[1],\sharp}  (x_2,\tilde \xi_2,x_3,\tilde \xi_3; h) = h^{-2}\underline{\mathsf N}_{h}^{[1],\sharp}(x_2,\tilde \xi_2,x_3,\tilde \xi_3)$. We define now a space of functions suitable for the
Birkhoff normal form in $(x_3, \tilde \xi_3, h)$. Let us now use the notation of the Appendix introduced in \eqref{eq.Cm} in the case when the family of smooth linear maps $\R^2\to\R^2$ is given by  
$$\varphi_{\mu, \R^2}(x_{2},\tilde\xi_{2})=(x_{2},\mu\tilde\xi_{2})\,.$$
Let
\[
\mathscr{F} := \mathcal{C}(1)_{\R^2},
\]
where the index $\R^2$ means that we consider symbols on $\R^2$.
More explicitly, we have
\[
\mathscr{F}=\{d\text{ s. t. } \exists c \in S(1; [0,1]\times(0,1])_{\R^2} :
d(x_2,\tilde \xi_2 ; \mu, h) =c(\varphi_{\mu,\R^2}(x_{2},\tilde\xi_{2}) ; \mu, h)\}\,.
\]
Then we define
\[
\mathscr{E} := \mathscr{F}[\![x_3,\tilde \xi_3, h]\!]\,,
\]
endowed with the full Poisson bracket 
\[
\mathscr{E}\times \mathscr{E} \ni (f,g) \mapsto \{f,g\} = \sum_{j=2,3}\deriv{f}{\tilde
  \xi_j}\deriv{g}{x_j} - \deriv{g}{\tilde \xi_j}\deriv{f}{x_j} \in
\mathscr{E},
\]
and the corresponding Moyal bracket $[f,g]$. We remark that the formal
Taylor series of the symbol $ \underline{\mathsf N}_{h}^{[1],\sharp} (x_2,\tilde \xi_2,x_3,\tilde \xi_3;\mu)$ with respect to $(x_3,\tilde\xi_3, h)$
belongs to $\mathscr{E}$.  We may apply the semiclassical Birkhoff
normal form relative to the main term $\underline{\nu}^2(x_2,\mu \tilde
\xi_2)(\tilde\xi_3^2 + x_3^2) $ exactly as in Section \ref{subsec.Birkhoff} (and also \cite[Proposition 2.7]{RVN13}), where we use the fact that the function
\[
(x_2,\tilde \xi_2,x_3,\tilde \xi_3 ;
   \mu,  h) \mapsto (\underline{\nu}^2(x_2,\mu \tilde \xi_2))^{-1}
\]
belongs to $\mathscr{E}$ because $\underline{\nu}^2>C>0$ uniformly
with respect to $\mu$. Let us consider $\gamma\in\mathscr{E}$ the
formal Taylor expansion of $h^{-2}\underline{r}_{h^2}^{\sharp}(x_2,\mu
\tilde \xi_2,x_3, h\tilde \xi_3)$ with respect to $(x_3,\tilde
\xi_3,h)$. The series $\gamma$ is of valuation $3$ and we obtain two
formal series $\kappa,\tau\in\mathscr{E}$ of valuation at least $3$
such that
\[
[\kappa, x_3^2 +\tilde \xi_3^2] =0
\]
and
\[
e^{ih^{-1} \textup{ad}_\tau}(\underline{\nu}^2(x_2,\mu \tilde
\xi_2)(\tilde\xi_3^2 + x_3^2) + \gamma) = \underline{\nu}^2(x_2,\mu
\tilde \xi_2)(\tilde\xi_3^2 + x_3^2) + \kappa.
\]
The coefficients of $\tau$ are in $S(1)$ and one can find a smooth
function $\tau_{h}\in S(1)$ with compact support with respect to $(x_{3},
\tilde\xi_{3}, h)$ and whose Taylor series in $(x_{3}, \tilde\xi_{3},
h)$ is $\tau$. By the Borel summation, $\tau_h$ will actually lie in $S(m')$ with
$m'(x_2,\tilde\xi_2, x_3, \tilde \xi_3)=\langle
(x_3,\tilde\xi_3)\rangle ^{-k}$ for any $k>0$, uniformly for small
$h>0$ and $\mu\in [0,1]$. Notice that $ \underline{\mathsf
  N}_{h}^{[1],\sharp} \in \mathcal{C}(m)$ with $m=\langle
(x_3,\tilde\xi_3)\rangle^2\geq 1$, and that $m m' = \OO(1)$.  

Then, we can apply Theorem \ref{theo:egorov3} with the family of
endomorphisms of $\R^4$ defined 
$$\varphi_{\mu, \R^4}(x_2,\tilde \xi_2, x_3, \tilde
\xi_3)= (x_2,\mu\tilde\xi_2, x_3, \tilde \xi_3)\,.$$ 
Thus, the new
operator
\[
\mathfrak{M}_{h}=e^{ih^{-1}\Op_{h}^w \tau_{h}} \underline{\mathfrak{
    N}}^{[1],\sharp}_{h} e^{-ih^{-1}\Op_{h}^w \tau_{h}}
\]
is a pseudo-differential operator whose Weyl symbol belongs to the
class $\mathcal{C}(m)$ modulo $h^\infty S(1)$ (see the notations of Theorem \ref{normal-form2}). Moreover, thanks to
Theorem \ref{theo:egorov4}, its symbol $\mathsf{M}_{h}$ admits the
following Taylor expansion (with respect to $(x_{3},
\tilde\xi_{3},h)$)
$$\tilde b(x_2,\mu\tilde \xi_2,s(x_2,\mu\tilde\xi_2))+\underline{\nu}^2(x_2,\mu\tilde \xi_2)(\tilde\xi_3^2 + x_3^2) + \kappa.$$
We write $\kappa=\sum_{m+2\ell\geq 3} c_{m,\ell}(x_{2}, \mu\tilde \xi_{2}) |\tilde z_{3}|^{\star 2m} h^\ell$ and we may find a smooth function $g^\star(x_{2},\mu\tilde\xi_{2}, Z, h)$ such that its Taylor series with respect to $Z$, $h$ is 
$$\sum_{2m+2\ell\geq 3} c_{m,\ell}(x_{2}, \mu\tilde \xi_{2}) Z^{m}  h^\ell.$$
We may now replace $\mu$ by $h$, which achieves the proof of Theorem
\ref{normal-form2}.

\subsection{Spectral reduction to the second normal form}
This section is devoted to the proof of Corollary \ref{normal-form2-d}.
\subsubsection{From $\mathcal{N}^{[1],\sharp}_{\hbar}$ to $\underline{\mathcal{N}}^{[1],\sharp}_{\hbar}$}
In this section, we prove Corollary \ref{pre-normal-form2-d}.
\begin{lemma}\label{number-N-uN}
We have
$$\mathsf{N}\left(\mathcal{N}^{[1],\sharp}_{\hbar},\beta_{0}\hbar\right)=\OO(\hbar^{-2}),\qquad \mathsf{N}\left(\underline{\mathcal{N}}^{[1],\sharp}_{\hbar},\beta_{0}\hbar\right)=\OO(\hbar^{-2}).$$
\end{lemma}
\begin{proof}
The first estimate comes from Proposition \ref{finite-k} and Corollary \ref{numbers-L-N}. The second estimate can be obtained by the same method as in the proof of Corollary \ref{numbers-L-N}.
\end{proof}
Let us now summarize the microlocalization properties of the eigenfunctions of $\underline{\mathcal{N}}^{[1],\sharp}_{\hbar}$ in the following proposition.
\begin{proposition}\label{micro-uN}
Let $\chi_{0}$ be a smooth cutoff function on $\R$ being $0$ in a neighborhood of $0$ and let $\delta\in\left(0,\frac{1}{2}\right)$. Let $\chi$ be a smooth cutoff function being $0$ on the bounded set $\{x_{3}^2+\underline{b}(x_{2},\xi_{2},s(x_{2},\xi_{2}))\leq\beta_{0}\}$ and $1$ on the set $\{x_{3}^2+\underline{b}(x_{2},\xi_{2},s(x_{2},\xi_{2}))\geq\beta_{0}+\tilde\eps\}$, with $\tilde\eps>0$. If $\lambda$ is an eigenvalue of $\underline{\mathcal{N}}^{[1],\sharp}_{\hbar}$ such that $\lambda\leq \beta_{0}\hbar$ and if $\psi$ is an associated eigenfunction, then we have
$$\Op_{\hbar}^w\left(\chi(x_{2},\xi_{2},x_{3})\right)\psi=\OO(\hbar^{\infty})\|\psi\|,$$
and 
$$\Op_{\hbar}^w\left(\chi_{0}\left(\hbar^{-\delta}\xi_{3}\right)\right)\psi=\OO(\hbar^{\infty})\|\psi\|.$$
\end{proposition}
\begin{proof}
The proof follows exactly the same lines as for Propositions \ref{space-loc-N} and \ref{xi3-loc-N}.
\end{proof}
Lemma \ref{number-N-uN} and Proposition \ref{micro-uN} on the one hand and Propositions \ref{space-loc-N} and \ref{xi3-loc-N} on the other hand are enough to deduce from Theorem \ref{pre-normal-form2} the point \eqref{pre-normal-form2-d-a} in Corollary \ref{pre-normal-form2-d}. The point \eqref{pre-normal-form2-d-b} easily follows from Corollary \ref{normal-form1-d}.

\subsubsection{From $\underline{\mathfrak{N}}^{[1],\sharp}_{h}$ to $\mathfrak{M}^\sharp_{h}$}
Let us now prove the point \eqref{normal-form2-d-a} in Corollary \ref{normal-form2-d}. We get the following rough estimate of the number of eigenvalues.
\begin{lemma}\label{number-N-uN'}
We have
\begin{equation}\label{eq.N-uN'1}
\mathsf{N}\left(\underline{\mathfrak{N}}^{[1],\sharp}_{h},\beta_{0}h^2\right)=\mathsf{N}\left(\mathfrak{M}_{h},\beta_{0}h^2\right)=\OO(h^{-4})\,,
\end{equation}
\begin{equation}\label{eq.N-uN'2}
\mathsf{N}\left(\mathfrak{M}^\sharp_{h},\beta_{0}h^2\right)=\OO(h^{-4})\,.
\end{equation}
\end{lemma}
\begin{proof}
First, we notice that $\underline{\mathfrak{N}}^{[1],\sharp}_{h}$ and $\mathfrak{M}_{h}$ are unitarily equivalent so that \eqref{eq.N-uN'1} holds. Then, given $\eta>0$ and $h$ small enough and up to shrinking the support of $g^\star$ and by using the Calderon-Vaillancourt theorem (as in the proof of Lemma \ref{lem.spess}), $\mathfrak{M}^\sharp_{h}\geq\widetilde{\mathfrak{M}}^\sharp_{h}$ in the sense of quadratic forms,
with
$$\widetilde{\mathfrak{M}}^\sharp_{h}=\Op_{h}^w\left(h^2 \underline{b}(x_2,h \tilde\xi_2,s(x_2, h\tilde\xi_2))\right)+h^2\mathcal{J}_{h}\Op_{h}^w\left(\left(\underline{\nu}^2(x_2,h\tilde\xi_2)\right)-\eta\right)\,.$$
Since $\underline{\nu}^2\geq c>0$, we get
\begin{multline*}
\Op_{h}^w\left(h^2 \underline{b}(x_2,h \tilde\xi_2,s(x_2, h\tilde\xi_2))\right)+h^2\mathcal{J}_{h}\Op_{h}^w\left(\left(\underline{\nu}^2(x_2,h\tilde\xi_2)\right)-\eta\right)\\
\geq \Op_{h}^w\left(h^2 \underline{b}(x_2,h \tilde\xi_2,s(x_2, h\tilde\xi_2))\right)+\frac{c}{2}h^2\mathcal{J}_{h}\,.
\end{multline*}
We deduce the upper bound \eqref{eq.N-uN'2} by separation of variables and the min-max principle.
\end{proof}

The following proposition deals with the microlocal properties of the eigenfunctions of $\underline{\mathfrak{N}}^{[1],\sharp}_{h}$.
\begin{proposition}\label{microloc-gN}
Let $\eta\in(0,1), \delta\in\left(0,\frac{\eta}{2}\right), C>0$. Let $\chi$ be a smooth cutoff function being $0$ on $\{\underline{b}(x_{2},\xi_{2},s(x_{2},\xi_{2}))\leq\beta_{0}\}$ and being $1$ on the set $\{\underline{b}(x_{2},\xi_{2},s(x_{2},\xi_{2}))\geq\beta_{0}+\tilde\eps\}$, with $\tilde\eps>0$. Let also $\chi_{1}$ be a smooth cutoff function on $\R^2$, being $0$ in a neighborhood of $0$.

If $\lambda$ is an eigenvalue of $\underline{\mathfrak{N}}^{[1],\sharp}_{h}$ such that $\lambda\leq \beta_{0}h^2$ and if $\psi$ is an associated eigenfunction, we have
\begin{equation}\label{chix2htxi2}
\Op_{h}^w\left(\chi(x_{2},h\tilde\xi_{2})\right)\psi=\OO(h^{\infty})\|\psi\|
\end{equation}
and if $\lambda$ is an eigenvalue of $\underline{\mathfrak{N}}^{[1],\sharp}_{h}$ such that $\lambda\leq b_{0}h^2+Ch^{2+\eta}$ and if $\psi$ is an associated eigenfunction, we have
\begin{equation}\label{chix3txi3}
\Op_{h}^w\left(\chi_{1}(h^{-\delta}(x_{3},\tilde\xi_{3}))\right)\psi=\OO(h^{\infty})\|\psi\|.
\end{equation}
\end{proposition}
\begin{proof}
The estimate \eqref{chix2htxi2} is a consequence of Proposition \ref{micro-uN}. Then, let us write the symbol of $\underline{\mathfrak{N}}^{[1],\sharp}_{h}$:
$$\underline{\mathsf{N}}^{[1],\sharp}_{h}=h^2\underline{\nu}^2(x_2,h\tilde\xi_2)\left(\tilde\xi_3^2 +x_3^2\right)+ h^2 \underline{b}(x_2,h\tilde\xi_2,s(x_2,h\tilde\xi_2))+\underline{R}^\sharp_{h^2}(x_{2},h\tilde\xi_{2}, x_{3},h\tilde\xi_{3}).$$
We write
\begin{multline*}
\left\langle\underline{\mathfrak{N}}^{[1],\sharp}_{h}\Op_{h}^w\left(\chi_{1}(h^{-\delta}(x_{3},\tilde\xi_{3}))\right)\psi,\Op_{h}^w\left(\chi_{1}(h^{-\delta}(x_{3},\tilde\xi_{3}))\right)\right\rangle\\
=\lambda\|\Op_{h}^w\left(\chi_{1}(h^{-\delta}(x_{3},\tilde\xi_{3}))\right)\psi\|^2\\
+\left\langle\left[\underline{\mathfrak{N}}^{[1],\sharp}_{h},\Op_{h}^w\left(\chi_{1}(h^{-\delta}(x_{3},\tilde\xi_{3}))\right)\right],\Op_{h}^w\left(\chi_{1}(h^{-\delta}(x_{3},\tilde\xi_{3}))\right)\psi\right\rangle.
\end{multline*}
We get
\begin{multline*}
\left\langle\left[\underline{\mathfrak{N}}^{[1],\sharp}_{h},\Op_{h}^w\left(\chi_{1}(h^{-\delta}(x_{3},\tilde\xi_{3}))\right)\right],\Op_{h}^w\left(\chi_{1}(h^{-\delta}(x_{3},\tilde\xi_{3}))\right)\psi\right\rangle\\
\leq Ch^{3}\left\|\Op_{h}^w\left(\underline{\chi}_{1}(h^{-\delta}(x_{3},\tilde\xi_{3}))\right)\psi\right\|^2,
\end{multline*}
where we have used \eqref{chix2htxi2}.
Then, we use that 
$$\underline{b}(x_2,h\tilde\xi_2,s(x_2,h\tilde\xi_2))\geq b_0,\qquad \underline{\nu}^2(x_2,h\tilde\xi_2)\geq c_{0}>0,\qquad \lambda\leq b_{0}h^2+Ch^{2+\eta},$$
and the G\aa rding inequality to deduce
\begin{multline*}
h^2\left(Ch^{2\delta}-Ch^{\eta}\right)\left\|\Op_{h}^w\left(\chi_{1}(h^{-\delta}(x_{3},\tilde\xi_{3}))\right)\psi\right\|^2\\
\leq Ch^{3}\left\|\Op_{h}^w\left(\underline{\chi}_{1}(h^{-\delta}(x_{3},\tilde\xi_{3}))\right)\psi\right\|^2.
\end{multline*}
The desired estimate follows by an iteration argument.
\end{proof}

In the same way we can deal with $\mathfrak{M}^\sharp_{h}$.
\begin{proposition}\label{microloc-gM}
Let $\eta\in(0,1), \delta\in\left(0,\frac{\eta}{2}\right), C>0$. Let $\chi$ be a smooth cutoff function being $0$ on $\{\underline{b}(x_{2},\xi_{2},s(x_{2},\xi_{2}))\leq\beta_{0}\}$ and being $1$ on the set $\{\underline{b}(x_{2},\xi_{2},s(x_{2},\xi_{2}))\geq\beta_{0}+\tilde\eps\}$, with $\tilde\eps>0$. If $\lambda$ is an eigenvalue of $\mathfrak{M}^{\sharp}_{h}$ such that $\lambda\leq \beta_{0}h^2$ and if $\psi$ is an associated eigenfunction, we have
\begin{equation}\label{chix2htxi2-gM}
\Op_{h}^w\left(\chi(x_{2},h\tilde\xi_{2})\right)\psi=\OO(h^{\infty})\|\psi\|
\end{equation}
and if $\lambda$ is an eigenvalue of $\mathfrak{M}^{\sharp}_{h}$ such that $\lambda\leq b_{0}h^2+Ch^{2+\eta}$ and if $\psi$ is an associated eigenfunction, we have
\begin{equation}\label{chix3txi3-gM}
\Op_{h}^w\left(\chi_{1}(h^{-\delta}(x_{3},\tilde\xi_{3}))\right)\psi=\OO(h^{\infty})\|\psi\|.
\end{equation}
\end{proposition}
\begin{proof}
In order to get \eqref{chix2htxi2-gM}, it is enough to go back to the representation with semiclassical $\hbar$, that is $\mathfrak{M}^{\sharp}_{h}=\mathcal{M}^{\sharp}_{\hbar}$. Indeed the microlocal estimate follows by the same arguments as in Propositions \ref{space-loc-N} and \ref{xi3-loc-N}. Then, \eqref{chix3txi3-gM} follows as in Proposition \ref{microloc-gN}.
\end{proof}

Propositions \ref{microloc-gN} and \ref{microloc-gM} and Theorem \ref{normal-form2} standardly imply the point \eqref{normal-form2-d-a} in Corollary \ref{normal-form2-d}. 
\subsubsection{From $\mathfrak{M}^\sharp_{h}$ to $\mathfrak{M}^{[1],\sharp}_{h}$}
Let us now prove the point \eqref{normal-form2-d-b} in Corollary \ref{normal-form2-d}. Note that the point \eqref{normal-form2-d-c} is just a reformulation of \eqref{normal-form2-d-b}.

Let us consider the Hilbertian decomposition $\mathfrak{M}^\sharp_{h}=\bigoplus_{k\geq1}\mathfrak{M}^{[k],\sharp}_{h}$, where the symbol $\mathsf{M}_{h}^{[k], \sharp}$ of $\mathfrak{M}^{[k],\sharp}_{h}$ is
$$h^2 \underline{b}(x_2,h \tilde\xi_2,s(x_2, h\tilde\xi_2))+(2k-1)h^3\underline{\nu}^2(x_2,h\tilde\xi_2) + h^2 g^{\star}(h, (2k-1)h,x_{2},h\tilde\xi_{2}).$$
There exists $h_{0}>0$ such that for all $k\geq 1$ and $h\in(0,h_{0})$,
\begin{multline*}
\langle \mathfrak{M}_{h}^{[k], \sharp}\psi,\psi\rangle\\
\geq \langle\Op_{h}^w\left(h^2 \underline{b}(x_2,h \tilde\xi_2,s(x_2, h\tilde\xi_2))+(2k-1)h^3(\underline{\nu}^2(x_2,h\tilde\xi_2)-\eps)\right)\psi,\psi\rangle.
\end{multline*}
Since each eigenfunction of $\mathfrak{M}_{h}^{[k],\sharp}$ associated with an eigenvalue less than $\beta_{0}h^2$ provides an eigenfunction of $\mathfrak{M}^\sharp_{h}$, we infer that the eigenfunctions of $\mathfrak{M}_{h}^{[k], \sharp}$ are uniformly microlocalized in a $(x_{2},\xi_{2})$-neighborhood of $(0,0)$ as small as we want. Therefore, on the range of $\mathds{1}_{(-\infty,b_{0}h^2)}(\mathfrak{M}_{h}^{[k], \sharp})$, we have
\begin{multline*}
\langle \mathfrak{M}_{h}^{[k], \sharp}\psi,\psi\rangle\\
\geq \langle\Op_{h}^w\left(h^2 \underline{b}(x_2,h \tilde\xi_2,s(x_2, h\tilde\xi_2))+(2k-1)h^3(\nu^2(0,0)-2\eps)\right)\psi,\psi\rangle.
\end{multline*}
and, with the G\aa rding inequality in the $\hbar$-quantization, we get
$$\langle \mathfrak{M}_{h}^{[k], \sharp}\psi,\psi\rangle\geq \langle\Op_{h}^w\left(h^2b_{0}+(2k-1)h^3(\nu^2(0,0)-\eps)-Ch^4\right)\psi,\psi\rangle.$$
This implies the point \eqref{normal-form2-d-b} in Corollary \ref{normal-form2-d}.

\section{Third Birkhoff normal form}\label{BNF3}

\subsection{Birkhoff analysis of the first level}
In this section we prove Theorem \ref{normal-form3}.

We consider $\mathcal{M}^{[1],\sharp}_{\hbar}=\Op_{\hbar}^w\left(M^{[1],\sharp}_{\hbar}\right)$, with
$$M_{\hbar}^{[1],\sharp}=\hbar \underline{b}(x_2, \xi_2,s(x_2, \xi_2))+ \hbar^{\frac{3}{2}}\underline{\nu}^2(x_2,\xi_2) + \hbar g^{\star}(\hbar^{\frac{1}{2}},\hbar^{\frac{1}{2}},x_{2},\xi_{2}).$$
By using a Taylor expansion, we get, 
\begin{multline}
  M_{\hbar}^{[1],\sharp}=\hbar
  b_{0}+\frac{\hbar}2\mathsf{Hess}_{(0,0)}\underline{b}(x_{2},\xi_{2},
  s(x_2,
  \xi_2))+\hbar^{\frac{3}{2}}\nu^2(0,0)+cx_{2}\hbar^{\frac{3}{2}}+d\xi_{2}\hbar^{\frac{3}{2}}
  \\ + \hbar\OO((\hbar^{\frac{1}{2}},z_{2})^3),
\end{multline}
where $c=\partial_{x_{2}}\nu^2(0,0)$ and
$d=\partial_{\xi_{2}}\nu^2(0,0)$, and we have identified the Hessian
with its quadratic form in $(x_2,\xi_2)$.

Then, there exists a linear symplectic change of variables that diagonalizes the Hessian, so that, if $L_{\hbar}$ is the associated unitary transform, 
$$L_{\hbar}^*\mathcal{M}^{[1],\sharp}_{\hbar}L_{\hbar}=\Op_{\hbar}^w\left(\hat M_{\hbar}^{[1],\sharp}\right),$$
with
$$\hat M_{\hbar}^{[1],\sharp}=\hbar b_{0}+\frac{\hbar}2\dd(x_{2}^2+\xi_{2}^2)+\hbar^{\frac{3}{2}}\nu^2(0,0)+\hat cx_{2}\hbar^{\frac{3}{2}}+\hat d\xi_{2}\hbar^{\frac{3}{2}}+\hbar\OO((\hbar^{\frac{1}{2}},z_{2})^3),$$
where 
$$\theta=\sqrt{\det\mathsf{Hess}_{(0,0)}b(x_{2},\xi_{2},s(x_2,\xi_2))}\,.$$
Since $(\partial_{x_{3}}b)(x_{2},\xi_{2},s(x_2,\xi_2))=0$ and $(0,0)$ is a critical point of $s$, we notice that $\partial^2_{x_{2} x_{3}}b(0,0,0)=\partial^2_{\xi_{2} x_{3}}b(0,0,0)=0$. Thus 
$$\det\mathsf{Hess}_{(0,0,0)}b(0,0,0)=\theta^2 \partial^2_{x_{3}}b(0,0,0).$$
Using that $b$ is identified with $b\circ\chi$ (see Remarks \ref{rem.thm1} and \ref{rem.Tchi}), this provides the expression given in \eqref{equ:alpha}.

Note that $\hat c^2+\hat d^2=\|(\nabla_{x_{2},\xi_{2}}\nu^2)(0,0)\|^2$
since the symplectic transform is in fact a rotation. Moreover we have
$$\dd(x_{2}^2+\xi_{2}^2)+\hat cx_{2}\hbar^{\frac{1}{2}}+ 
\hat d\xi_{2}\hbar^{\frac{1}{2}}=\dd\left(\left(x_{2}- \frac{\hat c
      \hbar^{\frac{1}{2}}}{\dd}\right)^2+\left(\xi_{2}-\frac{\hat d
      \hbar^{\frac{1}{2}}}{\dd}\right)^2\right)-\hbar\frac{\hat
  c^2+\hat d^2}{\dd}.$$ Thus, there exists a unitary transform $\hat
U_{\hbar^{\frac{1}{2}}}$, which is in fact an $\hbar$-Fourier Integral
Operator whose phase admits a Taylor expansion in powers of
$\hbar^{\frac{1}{2}}$, such that
$$\hat U_{\hbar^{\frac{1}{2}}}^*L_{\hbar}^*\mathcal{M}^{[1],\sharp}_{\hbar}L_{\hbar}\hat U_{\hbar^{\frac{1}{2}}}=:\underline{\mathcal{F}}_{\hbar}=\Op_{\hbar}^w\left(\underline{F}_{\hbar}\right),$$
where
$$\underline{F}_{\hbar}= \hbar b_{0}+\hbar^{\frac{3}{2}}\nu^2(0,0)- \frac{\|(\nabla_{x_{2},\xi_{2}}\nu^2)(0,0)\|^2}{2\dd}\hbar^2+ \hbar\left(\frac{\dd}2 |z_{2}|^2+\OO((\hbar^{\frac{1}{2}},z_{2})^3)\right).$$
Now we perform a semiclassical Birkhoff normal form in the space of
formal series $\R\formel{x_{2},\xi_{2},\hbar^{\frac{1}{2}}}$ equipped
with the degree such that $x_{2}^\ell\xi_{2}^m\hbar^{\frac{n}{2}}$ is
$\ell+m+n$ and endowed with the Moyal product. Let
$\underline{F}_{\hbar}^T$ be the full Taylor series of
$\underline{F}_{\hbar}$. We find a formal series
$\tau(x_2,\xi_2,\h^{\frac{1}{2}})$ with a valuation at least $3$ such
that
\[
e^{i\h^{-1} \textup{ad}_\tau}   \underline{F}_{\hbar}^T = F_\h^T,
\]
where $F_\h^T$ is a formal series of the form
$$F_{\hbar}^T=\hbar b_{0}+\hbar^{\frac{3}{2}}\nu^2(0,0)-
\frac{\|(\nabla_{x_{2},\xi_{2}}\nu^2)(0,0)\|^2}{2\dd}\hbar^2+
\frac{\dd}2\hbar |z_{2}|^2+\hbar
k^T(\hbar^{\frac{1}{2}},|z_{2}|^2),$$
and $k^T$ is a formal series in $\R\formel{\hbar^{\frac{1}{2}},|z_{2}|^2}$ (and that can be also written as a formal series in Moyal power of $|z_{2}|^2$, say $(k^T)^\star$).

Let $\tilde\tau(x_2,\xi_2,\mu)$ be a compactly supported function
whose Taylor expansion at $(0,0,0)$ is equal to $\tau(x_2,\xi_2,\mu)$.
By the Egorov theorem~\ref{theo:egorov2}, uniformly with respect to
the parameter $\mu$, we obtain that
$$e^{-i\hbar^{-1}\Op_{\hbar}^w\left(\tilde \tau\right)}
\Op_\h^w(\underline{F}_{\mu^2})
e^{i\hbar^{-1}\Op_{\hbar}^w\left(\tilde\tau\right)}=:\Op_{\hbar}^w(\tilde
F_{\mu})$$ is an $\hbar$-pseudo-differential operator depending
smoothly on $\mu$. Expanding $\tilde F_\mu$ in powers of $\mu$ in the
$S(1)$ topology, and letting $\mu=\sqrt{\h}$, we see that $\tilde
F_{\sqrt\hbar} = F_{\hbar} + \tilde G_{\hbar}$, where
$$F_{\hbar}=\hbar b_{0}+\hbar^{\frac{3}{2}}\nu^2(0,0)-
\frac{\|(\nabla_{x_{2},\xi_{2}}\nu^2)(0,0)\|^2}{2\dd}\hbar^2+
\frac{\dd}2\hbar |z_{2}|^2+\hbar k(\hbar^{\frac{1}{2}},|z_{2}|^2),$$
with $k$ a smooth function with a support as small as desired
w.r.t. its second variable, and $\tilde G_h=\h\OO(\abs{z_2}^\infty)$. It remains to notice that $\Op_{\h}^w\left( k(\h^{\frac{1}{2}},|z_{2}|^2)\right)$ can be written as $k^{\star}(\h^{\frac{1}{2}},\mathcal{K}_{\h})$ modulo $\Op_{\h}^w\left(\mathcal{O}(|z_{2}|^\infty)\right)$. This achieves the proof of Theorem \ref{normal-form3}.

\subsection{Spectral reduction to the third normal form}
Corollary \ref{normal-form3-d} is a consequence of the following lemma and proposition.

\begin{lemma}
We have
$$\mathsf{N}\left(\mathcal{M}^{[1],\sharp}_{\hbar}, \beta_{0}\hbar\right)=\OO(\hbar^{-2}),\qquad \mathsf{N}\left(\mathcal{F}_{\hbar},b_{0}\hbar+C\hbar^{1+\eta}\right)=\OO(\hbar^{-1+\eta}).$$
\end{lemma}
\begin{proof}
The first estimate follows from Lemma \ref{number-N-uN'} and the second one from a comparison with the harmonic oscillator in $x_{2}$.
\end{proof}
The last proposition concerns the microlocalization of the eigenfunctions.
\begin{proposition}
Let $\eta\in(0,1), \delta\in\left(0,\frac{\eta}{2}\right), C>0$. Let $\chi$ be a smooth cutoff function being $0$ in a bounded neighborhood of $(0,0)$ and $1$ outside a bounded neighborhood of $(0,0)$. If $\lambda$ is an eigenvalue of $\mathcal{M}^{[1],\sharp}_{\hbar}$ or of $\mathcal{F}_{\hbar}$ such that $\lambda\leq b_{0}\hbar+C\hbar^{1+\eta}$ and if $\psi$ is an associated eigenfunction, we have
$$\Op_{\hbar}^w\left(\chi(\hbar^{-\delta}(x_{2},\xi_{2}))\right)\psi=\OO(\hbar^\infty).$$
\end{proposition}
\begin{proof}
The proof is similar with the one of Proposition \ref{microloc-gN}.
\end{proof}

\appendix
\section{Egorov theorems}\label{app.egorov}
We start with the classical result (see for instance \cite[Theorem
11.1]{Z13} and~\cite[Théorème IV.10]{Ro87}).
\begin{theo}[\mbox{\cite[Theorem 11.1, Remark (ii)]{Z13}}]
\label{theo:egorov1}
  Let $P$ and $Q$ be $h$-pseudo-differential operators on $\R^d$, with
  $P\in\Op_{h}^w\left( S(1)\right)$ and $Q\in \Op_{h}^w\left(S(1)\right)$.  Then the operator $e^{\frac{i}{h}Q} P
  e^{-\frac{i}{h}Q}$ is a pseudo-differential operator in $\Op_{h}^w\left(S(1)\right)$, and
  \[
e^{\frac{i}{h}Q} P e^{-\frac{i}{h}Q} - \Op^w_h(p\circ \kappa) \in
  h \Op_{h}^w\left(S(1)\right)\,.
\] 
Here $p$ is the Weyl symbol of $P$, and the canonical
transformation $\kappa$ is the time-1 Hamiltonian flow associated with
principal symbol of $Q$.
\end{theo}
From this classical version of Egorov's theorem, one can deduce the
following refinement that is useful when $p$ does not belong to $S(1)$
(as it is the case in this paper).
\begin{theo}
\label{theo:egorov2}
  Let $P$ and $Q$ be $h$-pseudo-differential operators on $\R^d$, with
  $P\in \Op_{h}^w\left(S(m)\right)$ and $Q\in \Op_{h}^w\left(S(m')\right)$, where $m$ and $m'$ are order functions
  such that:
  \begin{equation}
    m' = \OO(1); \quad m m' = \OO(1).\label{equ:egorov0}
\end{equation}

Then the operator $e^{\frac{i}{h}Q} P e^{-\frac{i}{h}Q}$ is a
pseudo-differential whose symbol is in $S(m)$, and $e^{\frac{i}{h}Q} P
e^{-\frac{i}{h}Q} - \Op^w_h(p\circ \kappa) \in h \Op_{h}^w\left(S(1)\right)$.
\end{theo}
\begin{proof}
The proof is based on the following observation.  In order to compare $\Op_h^w(p\circ \kappa^t)$ and
  $e^{\frac{it}{h}Q} P e^{-\frac{it}{h}Q}$, we consider the
  derivative:
\begin{multline*}
\frac{d}{d\tau} \left(e^{\frac{i\tau}{h}Q} \Op_h^w(p\circ
  \kappa^{t-\tau}) e^{-\frac{i\tau}{h}Q}\right)\\ 
  = e^{\frac{i\tau}{h}Q}
\left( \frac{i}{h}[Q, \Op_h^w(p\circ \kappa^{t-\tau})] +
  \frac{d}{d\tau} \Op_h^w(p\circ \kappa^{t-\tau})
\right)e^{-\frac{i\tau}{h}Q}.
\end{multline*}
From Hypothesis~\eqref{equ:egorov0}, the term $[Q, \Op_h^w(p\circ
\kappa^{t-\tau})]$ belongs to $\Op_{h}^w\left(S(1)\right)$; moreover, if we denote by $q_0$
the principal symbol of $Q$, we have
\[
\frac{d}{d\tau} \Op_h^w(p\circ \kappa^{t-\tau}) = - \Op_h^w
(\{q_0,p\circ \kappa^{t-\tau}\}),
\]
which implies that this term is also in $\Op_{h}^w\left(S(1)\right)$. By symbolic calculus, we see that
\begin{equation}
  \frac{i}{h}[Q, \Op_h^w(p\circ \kappa^{t-\tau})] + \frac{d}{d\tau}
  \Op_h^w(p\circ \kappa^{t-\tau}) \in h \Op_{h}^w\left(S(1)\right),
\label{equ:egorov-preuve-derive}
\end{equation}
uniformly for $t,\tau$ in compact sets.  It follows by integration
from $0$ to $t$ that
\begin{equation}
  e^{\frac{it}{h}Q} P e^{-\frac{it}{h}Q} = \Op_h^w(p\circ \kappa^t) +
  h \int_0^t e^{\frac{is}{h}Q} P_1(s) e^{-\frac{is}{h}Q} ds,
\label{equ:egorov-preuve-integral}
\end{equation}
for some $P_1(s)\in \Op_{h}^w\left(S(1)\right)$, uniformly for $s\in[0,t]$. Applying
Theorem~\ref{theo:egorov1} to the integrand, we see that
$e^{\frac{it}{h}Q} P e^{-\frac{it}{h}Q} - \Op_h^w(p\circ \kappa^t) \in
h \Op_{h}^w\left(S(1)\right)$.
\end{proof}
In order to quantize the formal Birkhoff procedure of Section
\ref{BNF2b}, one needs to consider symbols in a class $\mathcal{C}$
stable under the Moyal product. For that purpose we first define the families of symbols $S(m; [0,1]\times(0,1])$, that is of smooth functions $a : \R^{2d} \times [0,1]\times(0,1] \rightarrow \mathbb{C}$ such that, for any $\alpha\in\N^{2d}$,
there exists $C_\alpha$ such that, $\forall (z ; \mu, h)\in\R^{2d}\times [0,1]\times(0,1]$,
$$
|\partial_{z}^\alpha a (z ; \mu, h)| \leq C_\alpha m(z)\,
$$
and where $m$ is an order function on $\R^{2d}$. The pair $(\mu,h)$ is considered as a parameter.

Then, let $(\varphi_{\mu})_{\mu\in [0,1]}$ be a smooth family of linear maps $\R^{2d}\to \R^{2d}$ and define the following families of symbols on $\R^{2d}$ by
 \begin{multline}\label{eq.Cm}
\mathcal{C}(m)= \Big\{ a\in S(m; [0,1]\times(0,1]);\quad a (z;\mu, h) = \tilde a (\varphi_{\mu}(z); \mu, h)\\
                                  \text{ with } \tilde a\in S(m; [0,1]\times(0,1])\Big\}\,.
\end{multline}

\begin{theo}
  \label{theo:egorov3}
 Let $P$ and
  $Q$ be $h$-pseudo-differential operators on $\R^d$, with $P\in
  \Op_{h}^w\left(\mathcal{C}(m)\right)$ and $Q\in \Op_{h}^w\left(\mathcal{C}(m')\right)$, where $m$ and
  $m'$ are order functions such that:
  \begin{equation*}
 m\geq 1;\quad   m' = \OO(1); \quad m m' = \OO(1).\label{equ:egorov}
  \end{equation*}

Then $e^{\frac{i}{h}Q} P e^{-\frac{i}{h}Q} = \tilde P + R$, where
$\tilde P \in\Op_{h}^w\left( \mathcal{C}(m)\right)$,
$R\in h^\infty \Op_{h}^w\left(S(1)\right)$, and with $\tilde P - \Op^w_h(p\circ \kappa) \in h
\Op_{h}^w\left(\mathcal{C}(1)\right)$.
\end{theo}
\begin{proof}
  Since $\varphi_\mu$ is linear, one can see (using for
  instance~\cite[Theorem 4.17]{Z13}) that $\mathcal{C}$ is stable under
  the formal Moyal product, \emph{i.e.} for all order functions $m_1$
  and $m_2$, we have
  $$(\mathcal{C}(m_1)) \star
  (\mathcal{C}(m_2)) \subset \mathcal{C}(m_1m_2) + h^\infty S(1).$$

  Let $\kappa$ be the canonical transformation associated with
  $Q$. Then, since $m\geq 1$, we have $p\circ \kappa\in\mathcal{C}(m)$; indeed, if we write the
  Hamiltonian flow of $Q$ in terms of the variable $\tilde z =
  \varphi_\mu(z)$, we see from the linearity of $\varphi_\mu$ that the
  components of the transformed vector field belong to
  $\mathcal{C}(m')$. Therefore $\varphi_\mu\circ \kappa$ is of the form
  $\tilde\kappa_\mu \circ \varphi_\mu$, for some diffeomorphism
  $\tilde\kappa_\mu$ depending smoothly on $\mu$.

  Therefore, both terms in~\eqref{equ:egorov-preuve-derive} belong to
  $\Op_{h}^w\left(\mathcal{C}(1)\right)$. Applying this argument inductively
  in~\eqref{equ:egorov-preuve-integral}, we may write, for any $k>0$,
\begin{equation*}
  e^{\frac{i}{h}Q} P e^{-\frac{i}{h}Q} - \Op_h^w(p\circ \kappa) -
  (h \tilde P_1 + h^2\tilde P_2 + \cdots + h^k\tilde P_k) \in h^{k+1}\Op_{h}^w\left(S(1)\right),
\end{equation*}
with $\tilde P_j\in \Op_{h}^w\left(\mathcal{C}(1)\right)$. By a Borel summation in $h$,
parametrized by $\tilde z = \varphi_\mu(z)$, we can find a symbol
$\hat P\in \Op_{h}^w\left(\mathcal{C}(1)\right)$ such that we have the asymptotic
expansion in $\Op_{h}^w\left(S(1)\right)$:
\[
\hat P \sim h \tilde P_1 + h^2\tilde P_2 + \cdots 
\]
We conclude by letting $\tilde P = \Op_h^w(p\circ \kappa) + \hat P$.
\end{proof}

We will also need to examine how the Egorov theorem behaves with
respect to taking formal power series of symbols. For this, it is
convenient to introduce a filtration of $S(m)$.
\begin{theo}\label{theo:egorov4}
  Let $m$ be an order function on $\R^{2d}$, and let
  $(\OO_j)_{j\in\N}$ be a filtration of $S(m)$, \emph{i.e.}:
  \[
  \OO_0 = S(m), \quad \OO_{j+1}\subset \OO_j.
  \]
  Let $P=\Op_{h}^w p$ and $Q=\Op_{h}^w q$ be $h$-pseudo-differential operators on $\R^d$, with
  $p\in S(m)$ and $q\in S(m')$, where $m'$ is an order function such
  that $m'$ and $mm'$ are bounded.

  Assume that:
\begin{equation}
  \tfrac{i}{h}\operatorname{ad}_q (\OO_j) \subset \OO_{j+1}; \qquad \forall j\geq 0.
\label{equ:egorov-formel}
\end{equation}

Then for any $k\geq 0$, the Weyl symbol of the pseudo-differential
operator $e^{\frac{i}{h}Q} P e^{-\frac{i}{h}Q} - \sum_{j=0}^k
\frac{1}{j!}(\tfrac{i}{h}\operatorname{ad}_Q )^jP$ belongs to
$\Op_{h}^w\left(\OO_{k+1}\right)$. In other words, the series of
$\exp(\tfrac{i}{h}\operatorname{ad}_Q)P$ converges to
$e^{\frac{i}{h}Q} P e^{-\frac{i}{h}Q}$ for the filtration
$(\OO_j)_{j\in\N}$.
\end{theo}
\begin{proof}
By the Taylor formula, we can write
\[
e^{\frac{i}{h}Q} P e^{-\frac{i}{h}Q} = \sum_{j=0}^k
\frac{1}{j!}(\operatorname{ad}_{ih^{-1}Q} )^jP +
\frac{1}{k!}(\operatorname{ad}_{ih^{-1}Q})^{k+1} \int_0^1 (1-t)^k
e^{\frac{it}{h}Q} P e^{-\frac{it}{h}Q}dt.
\]
By Theorem~\ref{theo:egorov2}, we see that the integral belongs to
$\Op_{h}^w\left(S(m)\right)=\Op_{h}^w\left(\OO_0\right)$. Therefore, by Assumption~\eqref{equ:egorov-formel}, the
remainder in the Taylor formula lies in $\Op_{h}^w\left(\OO_{k+1}\right)$.
\end{proof}

\subsubsection*{Acknowledgments} The authors would like to thank Yves
Colin de Verdi\`ere for stimulating discussions. This work was
partially supported by the ANR (Agence Nationale de la Recherche),
project {\sc Nosevol} n$^{\rm o}$ ANR-11-BS01-0019 and by the Centre
Henri Lebesgue (program \enquote{Investissements d'avenir} -- n$^{\rm
  o}$ ANR-11-LABX-0020-01). Y. K. was partially supported by the
Russian Foundation of Basic Research, project 13-01-91052-NCNI-a. During the completion of this work, B. H. was Simons foundation visiting
fellow at the Isaac Newton Institute in Cambridge.

\def\cprime{$'$}

\end{document}

%% file: circ.tex
\begin{picture}(0,0)%
\includegraphics{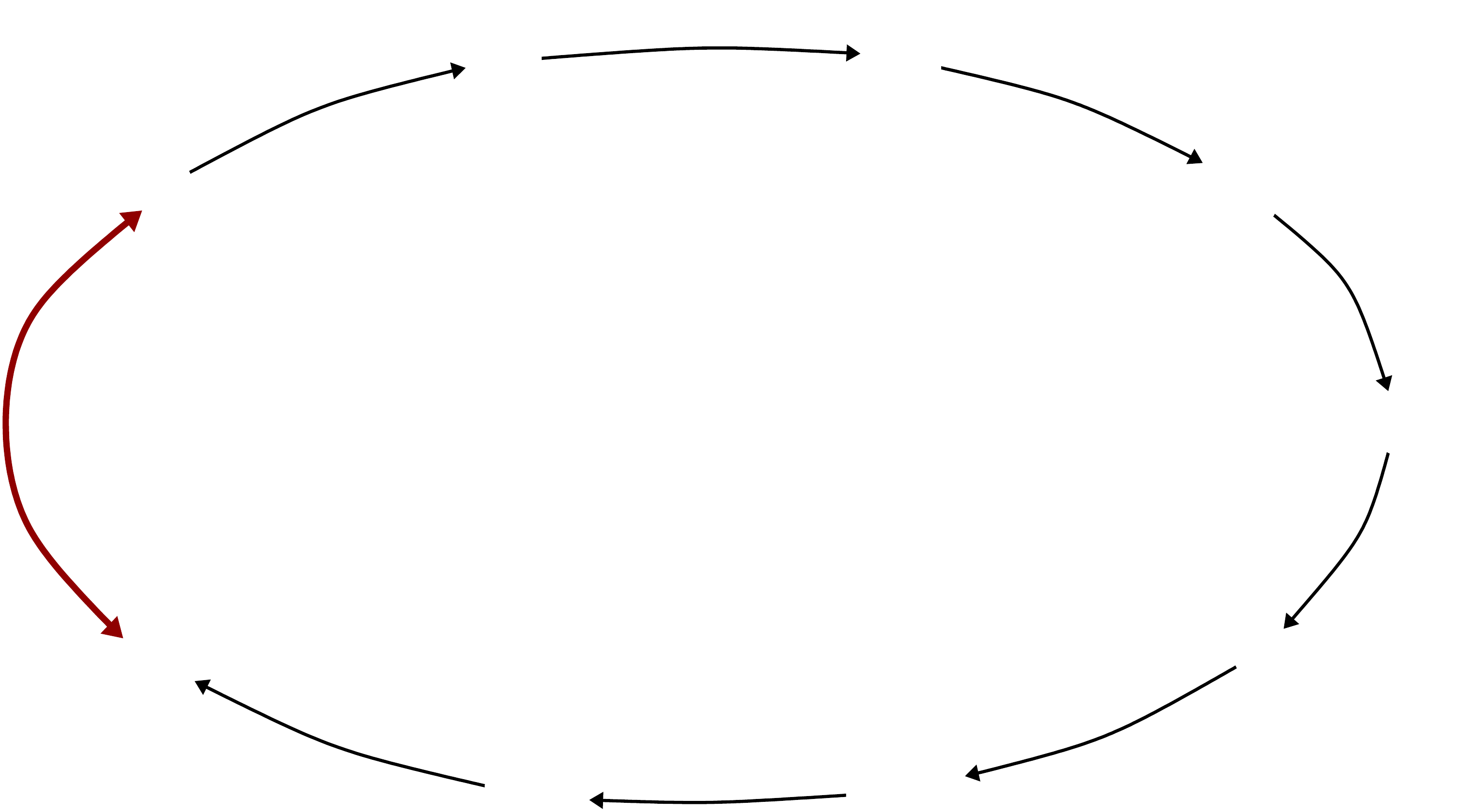}%
\end{picture}%
\setlength{\unitlength}{4144sp}%
\begingroup\makeatletter\ifx\SetFigFont\undefined%
\gdef\SetFigFont#1#2#3#4#5{%
  \reset@font\fontsize{#1}{#2pt}%
  \fontfamily{#3}\fontseries{#4}\fontshape{#5}%
  \selectfont}%
\fi\endgroup%
\begin{picture}(13996,7677)(95,-7822)
\put(451,-4291){\makebox(0,0)[lb]{\smash{{\SetFigFont{20}{24.0}{\rmdefault}{\mddefault}{\updefault}{\color[rgb]{0,0,0}Theorem \ref{t:main}}%
}}}}
\put(11611,-1996){\makebox(0,0)[lb]{\smash{{\SetFigFont{20}{24.0}{\rmdefault}{\mddefault}{\updefault}{\color[rgb]{0,0,0}$\underline{\mathcal{N}}_{\hbar}^{[1]}$}%
}}}}
\put(13141,-4201){\makebox(0,0)[lb]{\smash{{\SetFigFont{20}{24.0}{\rmdefault}{\mddefault}{\updefault}{\color[rgb]{0,0,0}$\underline{\mathfrak{N}}_{h}^{[1]}$}%
}}}}
\put(11926,-6406){\makebox(0,0)[lb]{\smash{{\SetFigFont{20}{24.0}{\rmdefault}{\mddefault}{\updefault}{\color[rgb]{0,0,0}$\mathfrak{M}_{ h}$}%
}}}}
\put(8506,-7666){\makebox(0,0)[lb]{\smash{{\SetFigFont{20}{24.0}{\rmdefault}{\mddefault}{\updefault}{\color[rgb]{0,0,0}$\mathfrak{M}^{[1]}_{h}$}%
}}}}
\put(4861,-7666){\makebox(0,0)[lb]{\smash{{\SetFigFont{20}{24.0}{\rmdefault}{\mddefault}{\updefault}{\color[rgb]{0,0,0}$\mathcal{M}^{[1]}_{\hbar}$}%
}}}}
\put(1531,-2041){\makebox(0,0)[lb]{\smash{{\SetFigFont{20}{24.0}{\rmdefault}{\mddefault}{\updefault}{\color[rgb]{0,0,0}$\mathcal{L}_{\hbar,\mathbf{A}}$}%
}}}}
\put(4636,-826){\makebox(0,0)[lb]{\smash{{\SetFigFont{20}{24.0}{\rmdefault}{\mddefault}{\updefault}{\color[rgb]{0,0,0}$\mathcal{N}_{\hbar}$}%
}}}}
\put(8461,-826){\makebox(0,0)[lb]{\smash{{\SetFigFont{20}{24.0}{\rmdefault}{\mddefault}{\updefault}{\color[rgb]{0,0,0}$\mathcal{N}_{\hbar}^{[1]}$}%
}}}}
\put(1441,-6496){\makebox(0,0)[lb]{\smash{{\SetFigFont{20}{24.0}{\rmdefault}{\mddefault}{\updefault}{\color[rgb]{0,0,0}$\mathcal{F}_{\hbar}$}%
}}}}
\put(13282,-1919){\rotatebox{300.0}{\makebox(0,0)[lb]{\smash{{\SetFigFont{17}{20.4}{\rmdefault}{\mddefault}{\updefault}{\color[rgb]{0,0,0}change of}%
}}}}}
\put(12972,-2098){\rotatebox{300.0}{\makebox(0,0)[lb]{\smash{{\SetFigFont{17}{20.4}{\rmdefault}{\mddefault}{\updefault}{\color[rgb]{0,0,0}semiclassical}%
}}}}}
\put(12661,-2276){\rotatebox{300.0}{\makebox(0,0)[lb]{\smash{{\SetFigFont{17}{20.4}{\rmdefault}{\mddefault}{\updefault}{\color[rgb]{0,0,0}parameter}%
}}}}}
\put(6324,-6797){\makebox(0,0)[lb]{\smash{{\SetFigFont{17}{20.4}{\rmdefault}{\mddefault}{\updefault}{\color[rgb]{0,0,0}change of}%
}}}}
\put(6324,-7155){\makebox(0,0)[lb]{\smash{{\SetFigFont{17}{20.4}{\rmdefault}{\mddefault}{\updefault}{\color[rgb]{0,0,0}semiclassical}%
}}}}
\put(6324,-7513){\makebox(0,0)[lb]{\smash{{\SetFigFont{17}{20.4}{\rmdefault}{\mddefault}{\updefault}{\color[rgb]{0,0,0}parameter}%
}}}}
\put(6031,-376){\makebox(0,0)[lb]{\smash{{\SetFigFont{17}{20.4}{\rmdefault}{\mddefault}{\updefault}{\color[rgb]{0,0,0}Corollary \ref{normal-form1-d}\ref{normal-form1-d-b}}%
}}}}
\put(2431,-1231){\rotatebox{20.0}{\makebox(0,0)[lb]{\smash{{\SetFigFont{17}{20.4}{\rmdefault}{\mddefault}{\updefault}{\color[rgb]{0,0,0}Theorem \ref{normal-form1}}%
}}}}}
\put(2656,-6811){\rotatebox{340.0}{\makebox(0,0)[lb]{\smash{{\SetFigFont{17}{20.4}{\rmdefault}{\mddefault}{\updefault}{\color[rgb]{0,0,0}Theorem \ref{normal-form3}}%
}}}}}
\put(12196,-5911){\rotatebox{60.0}{\makebox(0,0)[lb]{\smash{{\SetFigFont{17}{20.4}{\rmdefault}{\mddefault}{\updefault}{\color[rgb]{0,0,0}Theorem \ref{normal-form2}}%
}}}}}
\put(9631,-7171){\rotatebox{20.0}{\makebox(0,0)[lb]{\smash{{\SetFigFont{17}{20.4}{\rmdefault}{\mddefault}{\updefault}{\color[rgb]{0,0,0}Corollary \ref{normal-form2-d}\ref{normal-form2-d-b}}%
}}}}}
\put(9856,-781){\rotatebox{340.0}{\makebox(0,0)[lb]{\smash{{\SetFigFont{17}{20.4}{\rmdefault}{\mddefault}{\updefault}{\color[rgb]{0,0,0}Theorem \ref{pre-normal-form2}}%
}}}}}
\end{picture}%

%% file: HKRVN15.bbl
\begin{thebibliography}{10}

\bibitem{AHS}
{\sc J.~Avron, I.~Herbst, B.~Simon}.
\newblock Schr\"odinger operators with magnetic fields. {I}. {G}eneral
  interactions.
\newblock {\em Duke Math. J.} {\bf 45}(4) (1978) 847--883.

\bibitem{benettin-sempio}
{\sc G.~Benettin, P.~Sempio}.
\newblock Adiabatic invariants and trapping of a point charge in a strong
  nonuniform magnetic field.
\newblock {\em Nonlinearity} {\bf 7}(1) (1994) 281--303.

\bibitem{BHR14}
{\sc V.~Bonnaillie-No\"el, F.~H\'erau, N.~Raymond}.
\newblock {Magnetic WKB constructions}.
\newblock {\em Preprint}  (2014).

\bibitem{BO27}
{\sc M.~Born, R.~Oppenheimer}.
\newblock {Zur Quantentheorie der Molekeln}.
\newblock {\em Ann. Phys.} {\bf 84} (1927) 457--484.

\bibitem{VuCha08}
{\sc L.~Charles, S.~V{\~u}~Ng{\d{o}}c}.
\newblock Spectral asymptotics via the semiclassical {B}irkhoff normal form.
\newblock {\em Duke Math. J.} {\bf 143}(3) (2008) 463--511.

\bibitem{cheverry14}
{\sc C.~Cheverry}.
\newblock {Can one hear whistler waves ?}
\newblock preprint hal-00956458, 2014.

\bibitem{colin-moyennisation}
{\sc Y.~Colin~de Verdi{\`e}re}.
\newblock La m{\'e}thode de moyennisation en m{\'e}canique semi-classique.
\newblock In {\em Journ{\'e}es {EDP}}. CNRS, Saint Jean de Monts 1996.

\bibitem{CFKS87}
{\sc H.~L. Cycon, R.~G. Froese, W.~Kirsch, B.~Simon}.
\newblock {\em Schr\"odinger operators with application to quantum mechanics
  and global geometry}.
\newblock Texts and Monographs in Physics. Springer-Verlag, Berlin, study
  edition 1987.

\bibitem{DiSj99}
{\sc M.~Dimassi, J.~Sj{\"o}strand}.
\newblock {\em Spectral asymptotics in the semi-classical limit}, volume 268 of
  {\em London Mathematical Society Lecture Note Series}.
\newblock Cambridge University Press, Cambridge 1999.

\bibitem{Duf83}
{\sc A.~Dufresnoy}.
\newblock Un exemple de champ magn\'etique dans {${\bf R}^{\nu }$}.
\newblock {\em Duke Math. J.} {\bf 50}(3) (1983) 729--734.

\bibitem{FouHel10}
{\sc S.~Fournais, B.~Helffer}.
\newblock {\em Spectral methods in surface superconductivity}.
\newblock Progress in Nonlinear Differential Equations and their Applications,
  77. Birkh\"auser Boston Inc., Boston, MA 2010.

\bibitem{HelKo11}
{\sc B.~Helffer, Y.~A. Kordyukov}.
\newblock Semiclassical spectral asymptotics for a two-dimensional magnetic
  {S}chr\"odinger operator: the case of discrete wells.
\newblock In {\em Spectral theory and geometric analysis}, volume 535 of {\em
  Contemp. Math.}, pages 55--78. Amer. Math. Soc., Providence, RI 2011.

\bibitem{HelKo13}
{\sc B.~Helffer, Y.~A. Kordyukov}.
\newblock Eigenvalue estimates for a three-dimensional magnetic {S}chr\"odinger
  operator.
\newblock {\em Asymptot. Anal.} {\bf 82}(1-2) (2013) 65--89.

\bibitem{HelKo14}
{\sc B.~Helffer, Y.~A. Kordyukov}.
\newblock Semiclassical spectral asymptotics for a magnetic schr\"odinger
  operator with non-vanishing magnetic field.
\newblock In {\em Geometric Methods in Physics: XXXII Workshop, Bialowieza},
  Trends in Mathematics, pages 259--278. Birkh\"auser, Basel 2014.

\bibitem{HelKo13b}
{\sc B.~Helffer, Y.~A. Kordyukov}.
\newblock {Accurate semiclassical spectral asymptotics for a two-dimensional
  magnetic Schr\"odinger operator}.
\newblock {\em Annales Henri Poincar\'e (to appear)}  (2015).

\bibitem{HelMo96}
{\sc B.~Helffer, A.~Mohamed}.
\newblock Semiclassical analysis for the ground state energy of a
  {S}chr\"odinger operator with magnetic wells.
\newblock {\em J. Funct. Anal.} {\bf 138}(1) (1996) 40--81.

\bibitem{HelRo84}
{\sc B.~Helffer, D.~Robert}.
\newblock Puits de potentiel g\'en\'eralis\'es et asymptotique semi-classique.
\newblock {\em Ann. Inst. H. Poincar\'e Phys. Th\'eor.} {\bf 41}(3) (1984)
  291--331.

\bibitem{HelSj89}
{\sc B.~Helffer, J.~Sj{\"o}strand}.
\newblock Semiclassical analysis for {H}arper's equation. {III}. {C}antor
  structure of the spectrum.
\newblock {\em M\'em. Soc. Math. France (N.S.)} {\bf 39} (1989) 1--124.

\bibitem{I98}
{\sc V.~Ivrii}.
\newblock {\em Microlocal analysis and precise spectral asymptotics}.
\newblock Springer Monographs in Mathematics. Springer-Verlag, Berlin 1998.

\bibitem{Martinez07}
{\sc A.~Martinez}.
\newblock A general effective {H}amiltonian method.
\newblock {\em Atti Accad. Naz. Lincei Cl. Sci. Fis. Mat. Natur. Rend. Lincei
  (9) Mat. Appl.} {\bf 18}(3) (2007) 269--277.

\bibitem{Persson60}
{\sc A.~Persson}.
\newblock Bounds for the discrete part of the spectrum of a semi-bounded
  {S}chr\"odinger operator.
\newblock {\em Math. Scand.} {\bf 8} (1960) 143--153.

\bibitem{PoRay12}
{\sc N.~Popoff, N.~Raymond}.
\newblock When the 3{D} magnetic laplacian meets a curved edge in the
  semiclassical limit.
\newblock {\em SIAM J. Math. Anal.} {\bf 45}(4) (2013) 2354--2395.

\bibitem{Ray12}
{\sc N.~Raymond}.
\newblock Semiclassical 3{D} {N}eumann {L}aplacian with variable magnetic
  field: a toy model.
\newblock {\em Comm. Partial Differential Equations} {\bf 37}(9) (2012)
  1528--1552.

\bibitem{Ray14}
{\sc N.~Raymond}.
\newblock {\em {Little Magnetic Book}}.
\newblock arXiv: 1405.7912 2014.

\bibitem{RVN13}
{\sc N.~Raymond, S.~V\~u Ng\d{o}c}.
\newblock {Geometry and Spectrum in 2D Magnetic Wells}.
\newblock {\em {To appear in Annales de l'Institut Fourier}}  (2014).

\bibitem{Ro87}
{\sc D.~Robert}.
\newblock {\em Autour de l'approximation semi-classique}, volume~68 of {\em
  Progress in Mathematics}.
\newblock Birkh\"auser Boston Inc., Boston, MA 1987.

\bibitem{Vu06}
{\sc S.~V{\~u}~Ng{\d{o}}c}.
\newblock {\em Syst\`emes int\'egrables semi-classiques: du local au global},
  volume~22 of {\em Panoramas et Synth\`eses [Panoramas and Syntheses]}.
\newblock Soci\'et\'e Math\'ematique de France, Paris 2006.

\bibitem{Vu09}
{\sc S.~V{\~u}~Ng{\d{o}}c}.
\newblock Quantum {B}irkhoff normal forms and semiclassical analysis.
\newblock In {\em Noncommutativity and singularities}, volume~55 of {\em Adv.
  Stud. Pure Math.}, pages 99--116. Math. Soc. Japan, Tokyo 2009.

\bibitem{weinstein-symplectic}
{\sc A.~Weinstein}.
\newblock Symplectic manifolds and their lagrangian submanifolds.
\newblock {\em ADVAM2} {\bf 6} (1971) 329--346.

\bibitem{weinstein-maslov}
{\sc A.~Weinstein}.
\newblock Fourier integral operators, quantization, and the spectra of
  {R}iemannian manifolds.
\newblock In {\em G\'eom\'etrie symplectique et physique math\'ematique
  (Colloq. Internat. CNRS, No. 237, Aix-en-Provence, 1974)}, pages 289--298.
  \'Editions C.N.R.S., Paris 1975.

\bibitem{Z13}
{\sc M.~Zworski}.
\newblock {\em Semiclassical analysis}, volume 138 of {\em Graduate Studies in
  Mathematics}.
\newblock American Mathematical Society, Providence, RI 2012.

\end{thebibliography}
